\documentclass[12pt]{article}

\usepackage{mntlibs}
\usepackage{docmute}
\usepackage{url}

\newcounter{mntcomm}

\begin{document}
\include{def}
\title{Particle Mean Field Variational Bayes}
\author{Minh-Ngoc Tran\thanks{\textit{Discipline of Business Analytics, The University of Sydney Business School}}
\and Paco Tseng\thanks{\textit{ARC Centre for Data Analytics for Resources and Environments (DARE)}}
\and Robert Kohn\thanks{\textit{School of Economics, UNSW Business School}}
}

\maketitle
\begin{abstract}
Mean Field Variational Bayes (MFVB) is one of the most computationally efficient Bayesian inference methods. However, its use has been restricted to models with conjugate priors or those that allow analytical calculations. This paper proposes a novel particle-based MFVB approach that greatly expands the applicability of the MFVB method. We establish the theoretical basis of the new method by leveraging the connection between Wasserstein gradient flows and Langevin diffusion dynamics, and demonstrate the effectiveness of this approach using Bayesian logistic regression, stochastic volatility, and deep neural networks.

\end{abstract}
\section{Introduction}
The main challenge of Bayesian statistics is to conduct inference of a computationally intractable posterior distribution, $\pi(x) \propto \exp{\left(-U(x)\right)}$, $x\in\SetR^d$, generally only known up a normalising constant. To solve this problem, there are two main classes of computational methods that provide different approaches to approximate $\pi$. The first one is Markov chain Monte Carlo (MCMC) methods \citep{metropolis1953equation,hastings1970monte,robert1999monte}. For many years, MCMC has been the standard approach for Bayesian analysis because of its theoretical soundness. The method constructs a Markov chain to produce simulation consistent samples from the target distribution $\pi$. A general MCMC approach is the Metropolis-Hastings algorithm that generates a Markov chain by first generating a proposed state from a proposal distribution, then using an acceptance rule to decide whether to accept the proposal or stay at the current state \citep{robert1999monte}.   

Another, and often more efficient, class of MCMC methods is based on the Langevin dynamics
\begin{equation} \label{eq:Langevin dynamics}
    \d X_t = -\frac12\nabla U\left(X_t\right)\d t + \d B_t
\end{equation}
where $\left\{B_t\right\}_{t\geq 0}$ is the Brownian process on $\mathbb{R}^d$. This stochastic differential equation (SDE) characterises  the dynamics of the process $\{X_t\}_{t\geq 0}$
whose distribution, under some regularity conditions on the potential energy $U(x)$, converge to $\pi$ as the stationary distribution. 
In practice, however, it is necessary to work with a discretisation of the SDE in \eqref{eq:Langevin dynamics}, whose the distribution might not  converge to $\pi$ \citep{roberts1996exponential}.
The Metropolis-Hastings acceptance rule is then needed to correct for the error produced by the discretisation. This method is known as the Metropolis-adjusted Langevin algorithm (MALA) \citep{besag1994comments,roberts1996exponential}. 

The Metropolis-Hastings acceptance rule is necessary to guarantee for convergence of 
MCMC methods, but it also prevents the use of MCMC in big data and big model settings. 
This is because calculating the Metropolis-Hastings acceptance probability requires the full data (however, see, e.g. \cite{Quiroz:JASA2019} for speeding up MCMC using data subsampling),
and that this probability can easily get close to zero when the dimension $d$ is high.
Other limitations of MCMC methods include the need for a sufficient burn-in period for the generated Markov chain to be distributed as $\pi$,
and the absence of effective stopping criteria for checking convergence.
These limitations can be circumvented by the Variational Bayes method at the cost of some approximation accuracy.

Variational Bayes (VB) \citep{Waterhouse:NIPS1996,Attias:UAI1999,Blei:JASA2017} emerges as an alternative approach to inference in complex posterior distributions with large datasets. More recently, it grown in popularity due to its ability to scale up in terms of both the model complexity and data size.  
Different to MCMC, the VB method proposes a family of distributions $\mathcal{Q}$, called the variational family, and then identifies within $\mathcal{Q}$ the closest variational distribution to $\pi$ with respect to KL-divergence, i.e., 
\begin{equation} \label{eq:vb}
    q^* \in \arg\min_{q \in \mathcal{Q}} \left\{ \text{KL}\left(q || \pi\right)\right\}. 
\end{equation}
A representative example is Mean-Field VB (MFVB) (sec. \ref{sec:reviewMFVB}) which imposes a factorisation structure on the variational
distributions, i.e. $\mathcal{Q} = \left\{q=\prod q_i\right\}$, and $q^*$, if it exists and is unique, is called the mean-field approximation of $\pi$. 
An obvious limitation of MFVB is that it fails to capture the posterior dependence between the factorized blocks.
Despite this limitation, MFVB has been widely used in applications; see, e.g. \cite{Wand:BA2011}, \cite{Giordani:JCGS2013}, \cite{wang2013variational} and \cite{zhang2020theoretical}.
Implementation of MFVB relies on conjugate priors and the ability to calculate the associated expectations (sec. \ref{sec:reviewMFVB}).
As a result, it is challenging to apply standard MFVB for some simple models such as Bayesian logistic regression. \par

Our work aims at extending the scope of MFVB to makes it widely applicable by combining MFVB with the Langevin dynamics and circumventing the main issues of either method while maintaining the strengths of both, that is providing a scalable Bayesian inference algorithm with theoretical guarantees.
The new method, called Particle Mean field VB (PMFVB), leverages the Langevin dynamics to bypass the limitations in standard MFVB,
and employs the theory of Wasserstein gradient flows to establish its theoretical guarantee. 
Wasserstein gradient flows are a fundamental element of the Optimal Transport (OT) theory \citep{Ambrosio:OTbook,villani2009optimal}, that quantifies the dissimilarity between probability measures and introduces a differential structure into the space of probability measures. 
Inspired by fluid dynamics, \cite{Jordan:SIAM1998} introduce the concept of the gradient flow of a functional defined on this space,
which is a continuous curve of probability measures along which the functional is optimised.
It turns out that the gradient flow of KL-divergence functional is identical to the Langevin dynamics (sec. \ref{sec:Preliminaries}).
This important connection between Wasserstein gradient flows and  Langevin dynamics, and Stochastic Differential Equation (SDE) theory in general, provides the theoretical foundation for our PMFVB procedure.

\paradot{Our contribution} We study the KL-divergence functional on the space of factorised distributions $\mathcal{Q}$ equipped with the 2-Wasserstein distance, and show that 
the MFVB optimisation problem \eqref{eq:vb} has an unique solution $q^*$. We propose an algorithm for approximating the optimal mean field distribution $q^*$ by particles
that are moved by combining the classical MFVB framework with the Langevin diffusion.
We show that the distribution of these particles converges to $q^*$. We also study the posterior consistency of $q^*$ in terms of the data size.
The numerical performance of PMFVB is demonstrated using Bayesian logistic regression and Stochastic Volatility - the statistical models that classical MFVB methods, to the best of our knowledge, have not been successfully applied to. 
We then develop a variant of PMFVB for inference in Bayesian deep learning, in which modifications to standard PMFVB are introduced to make PMFVB suitable for deep neural networks.
For Bayesian inference in deep neural networks, Stochastic Gradient Langevin Dynamics (SGLD) of \cite{Welling.Teh:2011} is among the most commonly used methods.   
We discuss connections between PMFVB and SGLD, and numerically compare their performance using both simulated and real datasets.
Code implementing the examples is available at \url{https://github.com/VBayesLab/PMFVB}.

\paradot{Related work} Our work builds on recent advances in the Optimal Transport theory and Langevin dynamics.
The most closely related work to ours is \cite{Yao.Yang:2022MFVB-WGL} who, by focusing on a particular MFVB framework for statistical models with local latent variables,
combine Wasserstein gradient flows with MFVB for dealing with the intractability of the optimal MFVB factorised distribution.
Our PMFVB framework is more general than earlier works and can be applied to any statistical models including deep learning.
\cite{Galy-Fajou:2021Entropy} develop a particle-based Gaussian approximation that uses a flow of linear transformations for moving the particles; the resulting curve of Gaussian distributions can be viewed as an approximation of the Wasserstein gradient flow of the KL divergence functional.
\cite{Lambert:2022VI_WGL} study convergence results of Gaussian Variational Inference using the theory of Bures-Wasserstein gradient flow. They use particles to realise the flow of Gaussian approximations, and also extend Gaussian approximation to mixtures of Gaussians approximation. 

The variant of PMFVB for neural networks is related to SGLD and its variants \citep{Welling.Teh:2011,Chen:ICML2014,Li.et.al:2016,Kim.et.al:2022}.

\paradot{Notation} $\nabla f(x)$ denotes the gradient vector of scalar function $f$ defined on $\SetR^d$.
For a vector-valued function $v$ defined on $\SetR^d$, its divergence is $\nabla\cdot v(x) = \sum_j\frac{\partial v_j(x)}{\partial x_j}$.
$\Delta f(x)$ is the Laplacian of $f$, $\Delta f(x) = \nabla\cdot(\nabla f(x))=\sum_j\frac{\partial^2f(x)}{\partial x_j^2}$.  
For a generic set $\X \subset \mathds{R}^d$, we denote by $\mathcal P(\X)$ the set of probability measures on $\X$; and for a measure $q$, with some abuse of notation, we will denote by $q(\d x)$ and $q(x)$
the probability measure and density function, respectively.

\section{Preliminaries}\label{sec:Preliminaries}
This section collects the preliminaries on optimal transport theory and Langevin diffusion that are used in the paper.

\subsection{Wasserstein space and gradient flow} 
Consider a generic set $\X\subset\SetR^d$, let $\P_2(\mathcal X)$ be the set of absolutely continuous probability measures on $\X$ with finite second moments.
For any $p,q\in\mathcal P_2(\mathcal X)$, let 
\bea
W_2(p,q)&=&\min_{\gamma\in\Gamma(p,q)}\Big\{\Big(\int_{\mathcal X\times\mathcal X}\|x-y\|^2\gamma(dx\times dy)\Big)^{1/2}\Big\}\label{eq: Was def 1}\\
&=&\min_{T:T_{\#}p=q}\Big\{\Big(\int_{\mathcal X}\|x-T(x)\|^2p(dx)\Big)^{1/2}\Big\}\label{eq: Was def 2}
\eea
be the 2-Wasserstein metric on $\mathcal P_2(\mathcal X)$.
Here, $\Gamma(p,q)$ denotes the set of joint probability measures on $\mathcal X\times\mathcal X$ with the marginals $p$ and $q$,
and $T_{\#}(p)$ is the push forward measure of $p$, i.e.
\[T_{\#}(p)(A)=p(x:T(x)\in A),\;\;\;A\subset\mathcal X.\]
The existence of \eqref{eq: Was def 1} and \eqref{eq: Was def 2} and their equivalence is well studied; see, e.g., \cite{Ambrosio:OTbook}.
It is well-known that, equipped with this metric, $\mathcal P_2(\mathcal X)$ becomes a metric space, often called the {\it Wasserstein space}, denoted by $\mathbb{W}_2(\mathcal X)$.  
The Wasserstein space $\mathbb{W}_2(\mathcal X)$ has many attractive properties \citep{Ambrosio:OTbook,villani2009optimal} that make it possible to perform calculus on this space. In particular, $\mathbb{W}_2(\mathcal X)$ can viewed as a Riemannian manifold \citep{Otto:CPDE2001} whose rich geometry structure can be exploited to efficiently solve optimisation problems such as \eqref{eq:vb}.

Consider the functional $F(q)=\KL(q\|\pi)$ defined on $\mathbb{W}_2(\mathcal X)$, with some fixed measure $\pi\in\mathbb{W}_2(\mathcal X)$.
\cite{Jordan:SIAM1998} propose the following iterative scheme, known as the JKO scheme, to optimize $F(q)$.
Let $q^{(0)}\in\mathbb{W}_2(\mathcal X)$ be some initial measure and $\epsilon>0$. At step $k\geq0$, define
\beq\label{eq: JKO}
q^{(k+1)}:=\arg\min_{q\in\mathbb{W}_2(\mathcal X)}\Big\{F(q)+\frac{1}{2\eps}W_2^2(q,q^{(k)})\Big\}.
\eeq
Denote by $\frac{\delta F}{\delta q}(q):\SetR^d\mapsto\SetR$ a first variation at $q$ of the functional $F$, i.e.
\beq\label{eq: first variation}
\lim_{\epsilon\to0}\frac{F(q+\epsilon\xi)-F(q)}{\epsilon}=\int \frac{\delta F}{\delta q}(q)\d\xi
\eeq
for all $\xi\in\mathbb{W}_2(\mathcal X)$ such that $F(q+\epsilon\xi)$ is defined. 
The first variation, defined up to an additive constant, characterises the change of $F$ at $q$. 
Let $v^{(k)}(x)=\nabla\frac{\delta F}{\delta q}(q^{(k)})(x):\SetR^d\mapsto\SetR^d$; from \eqref{eq: first variation}, it can be shown that
\beq\label{eq: velocity}
v^{(k)}(x) = \nabla\log\pi(x)-\nabla\log q^{(k)}(x). 
\eeq
\cite{Jordan:SIAM1998} prove that (see also \cite[Chapter 8]{Santambrogio:OTbook} and \cite[Chapter 10]{Ambrosio:OTbook}), as $\eps\to0$, the discrete-time solution $\{q^{(k)}\}_{k=0,1,...}$ from \eqref{eq: JKO} converges to the continuous-time
 solution $\{q_t\}_{t\geq0}$ of the {\it continuity equation}
 \beq\label{eq: continuity equation}
 \frac{\partial q_t(x)}{\partial t}+\nabla\cdot\big(q_t(x) v_t(x)\big)=0
 \eeq
with $v_t(x)=\nabla\frac{\delta F}{\delta q}(q_t)(x)=\nabla\log\pi(x)-\nabla\log q_t(x)$.
For $\psi_t(x)=\log(q_t(x)/\pi(x))$, by noting that $\int ({\d\log q_t(x)}/{\d t})q_t(x)dx=0$, we have
\bean
\frac{\d F(q_t)}{\d t}=\int\psi_t(x)\frac{\partial q_t(x)}{\partial t}dx&=&-\int\psi_t(x)\nabla\cdot(q_t(x) v_t(x))\d x\\
&=&\E_{q_t}<\nabla\psi_t, v_t>\\
&=&-\E_{q_t}\big(\|v_t(x)\|^2\big)<0,
\eean
which justifies that the curve $\{q_t\}_{t\geq0}$, called the {\it gradient flow}, minimizes the KL functional $F(q)$.

\subsection{Langevin Monte Carlo diffusion} \label{subsec:LMC}
Let $\pi(\d x)$ be a target probability measure defined on $\X\subset\SetR^d$ with density $\pi(x)$.
The Langevin diffusion is the stochastic process $\{L_t\}_{t\geq 0}$ governed by the SDE
\bea\label{eq: Langegin diffusion}
L_0 &\sim& p_0\notag \\ 
\d L_t&=&\frac{1}{2}\nabla\log\pi(L_t)\d t+\d B_t,\;\;\;t>0
\eea
where $B_t$ is the $d-$dimensional Brownian motion and $p_0$ is an initial distribution on $\X\subset\SetR^d$.
Under some regularity conditions on $U(x)=-\log\pi(x)$, this SDE has an unique solution which is an ergodic Markov process with the invariant distribution $\pi(\d x)$ \cite[Chapter 4]{Pavliotis:SDEbook}.

Let $q_t(x)$ be the probability density (w.r.t. the Lebesgue measure on $\SetR^d$) of $L_t$; then $\{q_t\}_{t\geq0}$ is the unique solution of the 
{\it Fokker–Planck equation} (often called the forward Kolmogorov equation in the probability literature)
\bea
\frac{\partial q_t(x)}{\partial t}&=& -\nabla\cdot\Big(q_t(x)\nabla\log\pi(x)\Big)+\Delta q_t(x),\;\;t>0\label{eq:FPE}\\
q_0(x)&=& p_0(x);
\eea
see \cite{Pavliotis:SDEbook}, Chapter 4.
One can easily check that 
\beqn
\nabla\cdot\Big(q_t(x)\nabla\log\pi(x)\Big)-\Delta q_t(x)=\nabla\cdot\big(q_t(x) v_t(x)\big),
\eeqn
with $v_t(x)=\nabla\log\pi(x)-\nabla\log q_t(x)$; hence, the Fokker–Planck equation \eqref{eq:FPE} is identical to the continuity equation in \eqref{eq: continuity equation}.
Therefore the curve $\{q_t\}_{t\geq0}$ induced by the Langevin dynamics can be viewed as a gradient flow that minimises some sort of a discrepancy between $q_t$ and $\pi$; see, \cite{Dalalyan:JRSSB2017} and \cite{Cheng:ALT2018}.

For a fixed $h>0$, consider the following Langevin Monte Carlo (LMC) diffusion which is a time-continuous discretization approximation of \eqref{eq: Langegin diffusion}
\beq\label{eq: Langegin MC diffusion}
\d X_t^h=\frac{1}{2}\nabla\log\pi(X_{\tau(t)}^h)\d t+\d B_t,\;\;\;t\geq0,
\eeq
where $\tau(t):=kh$ if $t\in[kh,(k+1)h)$. Equation \eqref{eq: Langegin MC diffusion} implies that, 
at the time points $kh$, $k=0,1,...$, we have
\beq\label{eq: Langegin MC diffusion 2}
X_{(k+1)h}^h=X_{kh}^h+\frac{h}{2}\nabla\log\pi(X_{kh}^h)+\sqrt{h}\eta_k, \;\;\eta_k\stackrel{iid}{\sim}N(0,I_d).
\eeq
Denote by $\mu_t^h$ the distribution of $X_t^h$, $t\geq0$, from the LMC diffusion \eqref{eq: Langegin MC diffusion}.
\cite{Cheng:ALT2018} prove the following lemma, which says that the functional $F(\cdot)$ is reduced along the curve $\{\mu_t^h\}_{t\geq0}$.

\begin{lemma}\label{lem:Cheng and Bartlett}[Cheng and Bartlett, 2018, Lemma 1]
Suppose that $U(x)=-\log\pi(x)$ is strongly convex and has a Lipschitz continuous gradient. That is, there exist constants $c>0$ and $C>0$ such that
\[cI_d\leq\nabla^2U(x)\leq CI_d,\;\;\;\text{ for all }x\in\X.\]
Let $F(\mu_t^h)=\KL(\mu_t^h\|\pi)$.
Then, when the step size $h$ is sufficiently small
\beq
\frac{\d F(\mu_t^h)}{\d t} \leq 0,\;\;\;t>0.
\eeq
\end{lemma}

Although the Langevin dynamics in \eqref{eq: Langegin diffusion} converges to the invariant distribution $\pi(\d x)$, its convergence rate is not optimal \citep{Pavliotis:SDEbook}. Many studies have aimed to improve the speed of convergence. A simple yet effective method is to add a momentum term to the drift coefficient $\nabla\log\pi(x)$; see, e.g., \cite{Hwang:AAP2005} and \cite{Kim.et.al:2022}.
We will use accelerated Langevin dynamics in our implementation of the PMFVB method in Section \ref{sec: PMFVB-NN}.


\section{Mean Field Variational Bayes}\label{sec:reviewMFVB}
We are concerned with the problem of approximating a target probability measure $\pi(\d x\times\d y)$ defined on $\Theta=\X\times\Y\subset\SetR^{d_x}\times\SetR^{d_y}$ with density $\pi(x,y)$ (with respect to some reference measure such as the Lebesgue measure). The methodology proposed in this paper can be easily extended to cases of more than two blocks, $\Theta=\X_1\times\X_2\times\cdots\times\X_k$, $k>2$.
MFVB approximates $\pi(\cdot)$ by a probability measure $q(\d x\times\d y)=q_x(\d x)q_y(\d y)$, where $q_x(\d x)\in\mathcal P(\X)$ and $q_y(\d y)\in\mathcal P(\Y)$.
Consider the following optimisation problem
\beq\label{eq: MFVB 1}
q^*\in\arg\min_{q\in \mathcal P(\X)\otimes\mathcal P(\Y)}\Big\{F(q)=\KL(q\|\pi)\Big\}.
\eeq
We study in Section \ref{sec: MFVB propoerties} conditions under which the problem in \eqref{eq: MFVB 1} is well-defined and has a unique solution.
Define
\beq\label{eq: optimal MFVB}
q_x^*(x)\propto\exp\Big(\E_{q_y}\big[\log\pi(x,y)\big]\Big),\;\;\;\;q_y^*(y)\propto\exp\Big(\E_{q_x}\big[\log\pi(x,y)\big]\Big).
\eeq
MFVB turns the optimisation problem \eqref{eq: MFVB 1} into the following coordinate-descent-type problem:
\beq\label{eq: MFVB 2}
\text{given $q_y$ solve: }\min_{q_x\in\mathcal P(\X)}\Big\{\KL(q_x\|q_x^*)\Big\},
\eeq
and
\beq\label{eq: MFVB 3}
\text{given $q_x$ solve: }\min_{q_y\in\mathcal P(\Y)}\Big\{\KL(q_y\|q_y^*)\Big\}.
\eeq
Assuming that $q_x^*(x)$ and $q_y^*(y)$ have a standard form
and that the expectations in \eqref{eq: optimal MFVB} can be computed,
the solutions in \eqref{eq: MFVB 2}-\eqref{eq: MFVB 3} are  $q_x^*(x)$ and $q_y^*(y)$ respectively.
These assumptions limit the use of MFVB to simple cases.
For example, even in the simple Bayesian logistic regression model,
MFVB cannot be used as the assumptions above are not satisfied.  
The next section proposes a method for solving \eqref{eq: MFVB 1} without making these assumptions, and Section \ref{sec: MFVB propoerties} studies the theoretical properties of $F(q)$ and its minimizer $q^*$.

\section{Particle Mean Field Variational Bayes}\label{sec:pMFVB}
We now present our particle MFVB procedure.
The key idea is that whenever the optimal solutions $q_x^*(x)$ and $q_y^*(y)$ 
of \eqref{eq: MFVB 2} and \eqref{eq: MFVB 3} are unavailable in closed form,  
we use Langevin Monte Carlo diffusions to iteratively approximate them.
We assume below that both $q_x^*(x)$ and $q_y^*(y)$ are intractable; in MFVB applications where the variational distribution is factorized into $K$ blocks, $q = q_1\times q_2\times  \cdots \times q_K $, one only needs to use Langevin Monte Carlo diffusions to approximate those optimal solutions $q_k^*$ that are intractable.
Particle MFVB works more efficiently when many but a few of the optimal $q_k^*$ are tractable.
Lemma \ref{lem:Cheng and Bartlett} suggests that the KL functional $F(q)$ decreases after each iteration,
together with the result that $F(q)$ has the unique solution (c.f. Corollary \ref{cor: unique solution}),
this justifies our particle MFVB procedure. Theorem \ref{the: Convergence of the particle MFVB} provides a formal proof. 

We use a set of particles $\{X_i^{(t)},i=1,...,M\}$ to approximate $q_x^*(x)$ at iteration $t$,
and $\{Y_i^{(t)},i=1,...,M\}$ to approximate $q_y^*(y)$, $t\geq 1$.
We note that, unlike Section \ref{sec:Preliminaries} where $t$ denotes  continuous time, in this section $t$ denotes the $t$th iteration in the PMFVB algorithm.
Given $q_y^{(t)}(\d y)$ which is approximated by $\wh q_y^{(t)}(\d y)=\frac{1}{M}\sum_{i}\delta_{Y_i^{(t)}}(\d y)$, a Langevin Monte Carlo diffusion is used to approximate $q_x^*(x)\propto\exp\Big(\E_{q_y^{(t)}}\big[\log\pi(x,y)\big]\Big)$:
\beq
X_i^{(t+1)}=X_i^{(t)}+\frac{h_x}{2}\E_{q_y^{(t)}}\big[\nabla_x\log\pi(X_{i}^{(t)},y)\big]+\sqrt{h_x}\eta_{x,i},\;\;\;i=1,...,M
\eeq
with $\eta_{x,i}\sim N_{d_x}(0,I)$. The term $\E_{q_y^{(t)}}\big[\nabla_x\log\pi(X_{i}^{(t)},y)\big]$ can be approximated using a subset of $\{Y_i^{(t)},i=1,...,M\}$ by
$\frac{1}{m}\sum_{k=1}^m\nabla_x\log\pi(X_{i}^{(t)},Y_{i_k}^{(t)})$,
where $\{Y_{i_k}^{(t)},k=1,...,m\}$ is a random subset of size $m$ from $\{Y_i^{(t)},i=1,...,M\}$.
That is,
\beq\label{eq: update X}
X_i^{(t+1)}=X_i^{(t)}+\frac{h_x}{2m}\sum_{k=1}^m\nabla_x\log\pi(X_{i}^{(t)},Y_{i_k}^{(t)})+\sqrt{h_x}\eta_{x,i},\;\;\;i=1,...,M.
\eeq
Similarly, given $q_x^{(t+1)}(\d x)$ approximated by the particles $\{X_i^{(t+1)},i=1,...,M\}$,
we use a Langevin Monte Carlo diffusion to approximate $q_y^*(y)\propto\exp\Big(\E_{q_x^{(t+1)}}\big[\log\pi(x,y)\big]\Big)$:
\beq\label{eq: update Y}
Y_i^{(t+1)}=Y_i^{(t)}+\frac{h_y}{2m}\sum_{k=1}^m\nabla_y\log\pi(X_{i_k}^{(t+1)},Y_{i}^{(t)})+\sqrt{h_y}\eta_{y,i},\;\;\;i=1,...,M
\eeq
with $\eta_{y,i}\sim N_{d_y}(0,I)$.
Algorithm \ref{Alg: particle MFVB} summarises the procedure.

\begin{algorithm}[H]
\caption{Particle MFVB}\label{Alg: particle MFVB}
\begin{algorithmic}[1]
\Procedure{PMFVB}{$M,\epsilon, q^0_{x}, q^0_{y}$}
\State Input: number of particles $M$, tolerance $\epsilon>0$, initial distributions $q^{(0)}_{x}$ and $q^{(0)}_{y}$.
\State Initialise $X_i \sim q^{(0)}_{x}$ and $Y_i \sim q^{(0)}_{y},~ i = 1,\dots M$. $t\leftarrow0$.
\While{\textbf{Not} $S(\epsilon)$}
    \State Update $\{X_i^{(t+1)},i=1,...,M\}$ as in \eqref{eq: update X}
    \State Update $\{Y_i^{(t+1)},i=1,...,M\}$ as in \eqref{eq: update Y}
    \State $t\leftarrow t+1$.
\EndWhile
\EndProcedure
\end{algorithmic}
\end{algorithm}
In Algorithm \ref{Alg: particle MFVB}, $S(\epsilon)$ denotes a stopping rule as a function of some tolerance $\epsilon>0$. The computational complexity in each iteration is $O(mM)$. Once the subset $\{Y_{i_k}^{(t)},k=1,...,m\}$ has been selected, the update in \eqref{eq: update X} can be paralellised across the particles $i$
as there is no communication required between the particles.
Similarly, the update in \eqref{eq: update Y} can also be parallelised.

We now discuss the stopping rule. When a validation data set is available, as is typically the case in deep learning applications,
one can use a stopping rule based on the validation error and stop 
the iteration in Algorithm \ref{Alg: particle MFVB} if the validation error, or a rolling window smoothed version of it, no longer decreases.
This stopping rule is recommended in Section \ref{sec: PMFVB-NN} where PMFVB is used for training deep neural networks.
Alternatively, the update in Algorithm \ref{Alg: particle MFVB} can be stopped using the lower bound.
Let $\pi(x,y)=\wt\pi(x,y)/C$ with $C$ the normalising constant. Then,
\beqn
\KL(q\|\pi)= -\mathcal{L}(q)+\log C,
\eeqn
where $\mathcal{L}(q)$ is the lower bound term 
\beqn
\mathcal{L}(q)=\int\log\wt\pi(x,y)q(\d x,\d y)+H(q),\;\;\;H(q)=-\int\log q(x,y)q(\d x,\d y).
\eeqn
The entropy term $H(q)$ encourages the spread of the particles to avoid their collapse to a degenerate distribution.
However, as the LMC diffusion already spreads the particles by adding a Gaussian noise to them (c.f. \eqref{eq: Langegin MC diffusion 2}), hence circumventing convergence to a degenerate measure. At the $t$th iteration, given the $M$ particles $\{X_i^{(t)},Y_i^{(t)}\}_{i=1}^M$ approximating $q$, we suggest approximating the entropy $H(q)$ by $-(1/M)\sum_{i=1}^M\log(1/M)=\log(M)$.
The lower bound term at the $t$th iteration is approximated by
\beq\label{eq: LB estimate}
\widehat{\mathcal{L}}=\frac{1}{M}\sum_i\log\wt\pi(X_i^{(t)},Y_i^{(t)})+\log(M).
\eeq

\section{Theoretical analysis of particle MFVB}\label{sec: Theoretical analysis}
In order to avoid technical complications, we will assume in this section that $\X$ and $\Y$ are compact sets in $\SetR^{d_x}$ and $\SetR^{d_y}$, respectively. All proofs of the theorems and corollaries are in the Appendix.

\subsection{Properties of functional $F(q)$ on the Wasserstein space}\label{sec: MFVB propoerties} 
We first study the properties of the KL functional $F(q)$ on the Wasserstein space $\mathcal Q=\mathbb{W}_2(\mathcal X)\otimes \mathbb{W}_2(\mathcal Y)$.
To the best of our knowledge, there is no previous work studying the theoretical properties of the MFVB problem \eqref{eq: MFVB 1}. By limiting $P(\X)\otimes\mathcal P(\Y)$ to the 
Wasserstein space $\mathbb{W}_2(\mathcal X)\otimes \mathbb{W}_2(\mathcal Y)$, the theorem below shows that the optimal MFVB distribution exists and is unique. 

\begin{theorem}\label{the: theorem 1} 
Assume that $\X$ and $\Y$ are compact sets and that $\pi(x,y)$ is continuous in both $x$ and $y$. Then,
\begin{itemize}
\item[(i)] $F(q)$ is lower semi-continuous (w.r.t. the weak convergence on $\Q$).
\item[(ii)] $F(q)$ is convex.
\end{itemize}
\end{theorem} 

\begin{corollary}\label{cor: unique solution}
Under the assumptions in Theorem \ref{the: theorem 1}, $F(q)$ has an unique minimizer on $\Q$. 
\end{corollary}

\subsection{Convergence of the particle MFVB algorithm}
As the number of particles $M\to\infty$, Algorithm \ref{Alg: particle MFVB} defines a sequence of measures $\{q^{(t)}(\d x\times\d y)=q_x^{(t)}(\d x)q_y^{(t)}(\d y),\ t=1,2,...\}$ on $\Q$.
The theorem below shows that the KL functional $F(q^{(t)})$ is non-increasing over $t$, and hence $q^{(t)}$ converges to the solution $q^*$.

\begin{theorem}\label{the: Convergence of the particle MFVB}
Assume that, for each fixed $y$, $U(x)=-\log\pi(x,y)$ is strongly convex and has a Lipschitz continuous gradient w.r.t. $x$. That is, there exist constants $c_1>0$ and $C_1>0$ such that
\[c_1I_{d_x}\leq\nabla^2U(x)\leq C_1I_{d_y},\;\;\;\text{ for all $x$}.\]
Similarly, for each fixed $x$, $V(y)=-\log\pi(x,y)$ is strongly convex and has a Lipschitz continuous gradient w.r.t. $y$. That is, there exist constants $c_2>0$ and $C_2>0$ such that
\[c_2I_{d_y}\leq\nabla^2V(y)\leq C_2I_{d_y},\;\;\;\text{ for all $y$}.\]
Then, if the step sizes $h_x$ and $h_y$ are sufficiently small and $M\to\infty$, $F(q^{(t)})$ is non-increasing over $t$, and $q^{(t)}$ converges to the unique minima $q^*$ of $F(q)$.
\end{theorem}

\subsection{Posterior consistency of $q^*$}
This section studies the properties of the particle MFVB solution $q^*=q_n^*$ in the context of Bayesian inference where the target $\pi(x,y)=\pi_n(\theta)$ is a posterior distribution
\[\pi_n(\theta)=p(\theta|X^{(n)})\]
of a model parameter $\theta\in\SetR^d$, and $X^{(n)}$ denotes the data of size $n$.
Is the PMFVB approximation $q_n^*$ posterior consistent? i.e. does $q_n^*(\d\theta)$ concentrate on a small neighborhood of
the true parameter as the sample size $n\to\infty$? We look at this problem from the frequentist point of view where we assume that the true
data generating parameter exists. To study this question, we first define some notation.

Let $p_\theta^{(n)}(\d X^{(n)})$ be the distribution of data $X^{(n)}$ under the model, parameterized by $\theta\in\Theta$.
Assume that $X^{(n)}$ is generated under some true parameter $\theta_0\in\Theta$, and denote by $p_0^{(n)}(\d X^{(n)})$ the true underlying distribution of $X^{(n)}$.
Let $\pi_0(\d\theta)\in\mathcal P(\Theta)$ be the prior distribution of $\theta$. The posterior distribution is
\beqn
\pi_n(\d\theta)\propto \pi_0(\d\theta)p_\theta^{(n)}(X^{(n)}).
\eeqn
We shall study the asymptotic behavior of the PMFVB variational posterior 
\beqn
q_n^*=\arg\min_{q\in\mathcal Q}\KL\big(q\|\pi_n\big)
\eeqn
with the space of probability measures $\mathcal Q$ having the factorised form in Section \ref{sec:reviewMFVB}, $\mathcal Q=\P(\X)\otimes\P(\Y)$, $\Theta=\X\times\Y$. We note that it is unnecessary to equip $\mathcal Q$ with the Wassterstein distance in this section.

The asymptotic behavior and convergence rate of the posterior $\pi_n(\d\theta)$ are well studied in the Bayesian statistics literature; see, e.g., \cite{ghosal2000}. The asymptotic behavior of the conventional VB approximation was studied recently in \cite{ZhangGao:2019} and \cite{Alquier.Ridgway:2019},
who study the conditions on the variational family (also the prior $\pi_0$ and the likelihood $p_\theta^{(n)}$) to characterize the convergence properties of the variational posterior.
The predicament with conventional VB is that the variational family should not be too large, so that the VB optimization is solvable;
but it should not be too small either, so that the variational posterior can still achieve some sort of posterior consistency.
It turns out that the variational family $\mathcal Q$ in our particle MFVB is general enough; using the results in \cite{ZhangGao:2019}, we will show that the PMFVB approximation $q_n^*$ enjoys the posterior consistency as the posterior $\pi_n(\d\theta)$ does.

We need the following assumptions on the prior $\pi_0$ and likelihood $p_\theta^{(n)}$.
\begin{assumption} Let $\varepsilon_n$ be a sequence of positive numbers such that $\varepsilon_n\to0$ and $n\varepsilon_n^2\geq1$, and $C_1,C_2$ and $C$ are constants such that  $C>C_2+2$.
\begin{itemize}
\item[(A1)] Testing condition: For any $\varepsilon>\varepsilon_n$, there exist a set $\Theta_n(\varepsilon)\subset\Theta$ and a test function $\phi_n=\phi_n(X^{(n)})\in[0,1]$ such that
\beqn\E_{p^{(n)}_0}(\phi_n)\leq\exp(-Cn\varepsilon^2)\;\;\text{and}\;\;\sup_{\substack{\theta\in\Theta_n(\varepsilon),\\ \|\theta-\theta_0\|_2^2>C_1\varepsilon^2}}\E_{p^{(n)}_\theta}(1-\phi_n)\leq\exp(-Cn\varepsilon^2)
\eeqn
\item[(A2)] Prior mass conditions:
\beqn
\pi_0\big(\Theta_n(\varepsilon)^c\big)\leq\exp(-Cn\varepsilon^2)
\eeqn
and, for some $\rho>1$,
\beqn
\pi_0\Big(\theta:D_\rho(p^{(n)}_0\|p^{(n)}_\theta)\leq C_2n\varepsilon_n^2\Big)>0,\;\;\;\text{for any $n$}.
\eeqn
\item[(A3)] Smoothness: 
\[|\log\pi_0(\theta)-\log\pi_0(\theta')|\leq C_3\|\theta-\theta'\|_2,\;\;\;\forall\theta,\theta'\in\Theta\]
for some $C_3>0$, and
\[|\log p_{\theta}^{(n)}(X^{(n)})-\log p_{\theta'}^{(n)}(X^{(n)})|\leq C_4(X^{(n)})\|\theta-\theta'\|_2,\;\;\;\forall\theta,\theta'\in\Theta\]
with $C_5:=\E_{p^{(n)}_0}\big[C_4(X^{(n)})\big]<\infty$. 
\end{itemize}
\end{assumption}
Here, $D_\rho(p\|q)$ denotes the $\rho$-R\'enyi divergence between two probability measures $p$ and $q$,
\beqn
D_\rho(p\|q) = \begin{cases}
\frac{1}{\rho-1}\log\int\left(\frac{\d p}{\d q}\right)^{\rho}\d q,&\rho\not=1,\\
\KL(p\|q),&\rho=1.
\end{cases}
\eeqn

Assumptions (A1) and (A2) are standard in the Bayesian statistics literature \citep{ghosal2000,ZhangGao:2019}, and are used to characterize the convergence rate of the posterior distribution.

Assumption (A1) states that, restricted to a subset $\Theta_n(\varepsilon)$ of $\Theta$, there exists a test function
that is able to distinguish the true probability measure from the complement of its neighborhood.
Assumption (A2) requires that the prior distribution concentrates on $\Theta_n(\varepsilon)$ and
puts a positive mass on ``good'' values of $\theta$ in the sense that $p^{(n)}_\theta$ is close enough to the true measure $p^{(n)}_0$ in terms of the $\rho$-R\'enyi divergence.
Under these assumptions, one can prove that the convergence rate of the posterior $\pi_n(\d\theta)$ to the true data generating measure is $\varepsilon_n^2$ \citep{ghosal2000,ZhangGao:2019}. Typically, $\varepsilon_n=1/\sqrt{n}$.  
Assumption (A3) requires some smoothness of the prior and likelihood with respect to $\theta$.
The theorem below shows that $q_n^*$ is posterior consistent.

\begin{theorem}\label{the: posterior consistency}
Suppose that conditions (A1), (A2) and (A3) are satisfied.
Then, for any $\epsilon>0$,
\beq
q_n^*\Big(\|\theta-\theta_0\|_2^2>\epsilon\Big)=o(1)\stackrel{n\to\infty}{\longrightarrow}0,\;\;\;\;p^{(n)}_0-a.s. 
\eeq
\end{theorem}


\section{Numerical examples}\label{sec: Numerical examples}

\subsection{Bayesian logistic regression}
Although Bayesian logistic regression is a benchmark model in statistics, is not straightforward to use the classical MFVB method here because of the lack of a conjugate prior.
This section demonstrates that it is straightforward to use the PMFVB method for Bayesian logistic regression. 
Consider the model
\beq\label{eq:BLR 1}
\beta\sim N(0,\sigma_0^2I_d),\;\;\;y_i\sim\text{Binomial}\big(1,\sigma(x_i^\top\beta)\big),\;\;\;
\sigma(x_i^\top\beta)=\frac{1}{1 +\exp(-x_i^\top\beta)},
\eeq
where $\beta=(\beta_0,\beta_1,\beta_2,\beta_3)^\top$ and $x_i=(1,x_{i,1},x_{i,2},x_{i,3})^\top$.
We generated a dataset of size $n=200$ from \eqref{eq:BLR 1} with $d=4$ and $\sigma_0^2=4$.
The likelihood is 
\[L(\beta)=\prod_{i=1}^n\big(\sigma(x_i^\top\beta)\big)^{y_i}\big(1-\big(\sigma(x_i^\top\beta)\big)\big)^{1-y_i}\]
with posterior 
\[\pi(\beta)\propto p(\beta)L(\beta).\]
We consider the PMFVB procedure with the model parameters $\beta$ factorized into two blocks $\theta_1=(\beta_0,\beta_1)^\top$ and $\theta_2=(\beta_2,\beta_3)^\top$. 
Write $x_i^{(1)}=(1,x_{i,1})^\top$ and $x_i^{(2)}=(x_{i,2},x_{i,3})^\top$.
The gradient of the log posterior with respect to $\theta_1$ and $\theta_2$ is given by
\beqn
\nabla_{\theta_1}\log\pi(\theta_1,\theta_2)=\sum_{i = 1}^n \left(y_i - \sigma\left({x_i^{(1)}}^\top\theta_1+{x_i^{(2)}}^\top\theta_2\right)\right)x_{i}^{(1)} - \frac{1}{\sigma_0^2}\theta_1
\eeqn
and
\beqn
\nabla_{\theta_2}\log\pi(\theta_1,\theta_2)=\sum_{i = 1}^n \left(y_i - \sigma\left({x_i^{(1)}}^\top\theta_1+{x_i^{(2)}}^\top\theta_2\right)\right)x_{i}^{(2)} - \frac{1}{\sigma_0^2}\theta_2,
\eeqn
respectively.

The PMFVB algorithm maintains a set of $M$ particles $\{\theta_{1,i}^{(t)},\theta_{2,i}^{(t)}\}_{i=1,...,M}$ over iterations $t\geq 0$. At iteration $t+1$, according to \eqref{eq: update X}, the first block of the particles is updated as
\beq\label{eq: PMFVB for BLR 1}
\theta_{1,i}^{(t+1)}=\theta_{1,i}^{(t)}+\frac{h}{2m}\sum_{k=1}^m\nabla_{\theta_1}\log\pi(\theta_{1,i}^{(t)},\theta_{2,i_k}^{(t)})+\sqrt{h}N(0,I_2),
\eeq
where $\{\theta_{2,i_k}^{(t)}\}_{k=1,...,m}$ is a random subset of size $m$ from $\{\theta_{2,i}^{(t)}\}_{i=1,...,M}$, and $h>0$ is a step size.
The second block of particles is updated as
\beq\label{eq: PMFVB for BLR 2}
\theta_{2,i}^{(t+1)}=\theta_{2,i}^{(t)}+\frac{h}{2m}\sum_{k=1}^m\nabla_{\theta_2}\log\pi(\theta_{1,i_k}^{(t+1)},\theta_{2,i}^{(t)})+\sqrt{h}N(0,I_2),
\eeq
where $\{\theta_{1,i_k}^{(t+1)}\}_{k=1,...,m}$ is a random subset from $\{\theta_{1,i}^{(t+1)}\}_{i=1,...,M}$.
The PMFVB procedure for approximating the posterior $\pi(\beta)$ iterates between \eqref{eq: PMFVB for BLR 1} and \eqref{eq: PMFVB for BLR 2} until stopping.

Below, we implemented the PMFVB algorithm using $M = 3000$ particles, and the optimisation time took $18$ seconds. In comparison, the MCMC (Halmitonian Monte Carlo) was conducted using PyMC with standard setup, 10,000 samples and 1,000 burn-in period, and the sampling procedure took 54 seconds. Both algorithms  were run on the same Dell Optiplex 7490 AIO (i7-11700) computer. The trace plot in Figure \ref{fig:bayeslasso} shows the lower bound \eqref{eq: LB estimate} over the iterations. 
\begin{figure}[h]
    \centering
    \includegraphics[width = 0.6\textwidth]{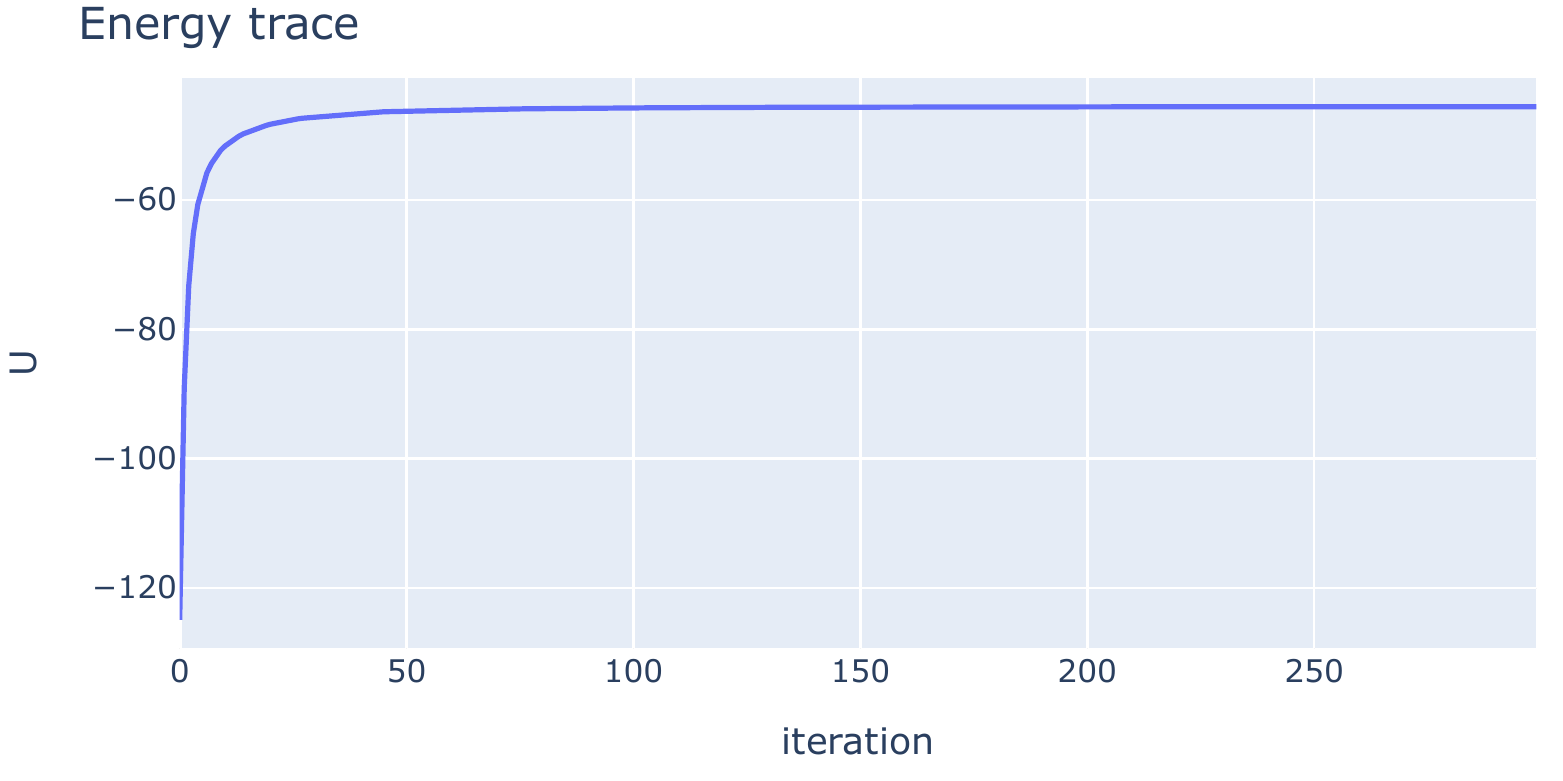}
    \caption{The trace plot of the PMFVB lower bound in Bayesian logistic regression.}
    \label{fig:bayeslasso}
\end{figure}
Figure \ref{fig:blogistic_hist} plots the marginal posteriors using kernel density estimation based on the PMFVB particles and Hamiltonian MC, and shows similar results. 
\begin{figure}[H]
\centering
\begin{subfigure}{.5\textwidth}
  \centering
  \includegraphics[width=\linewidth]{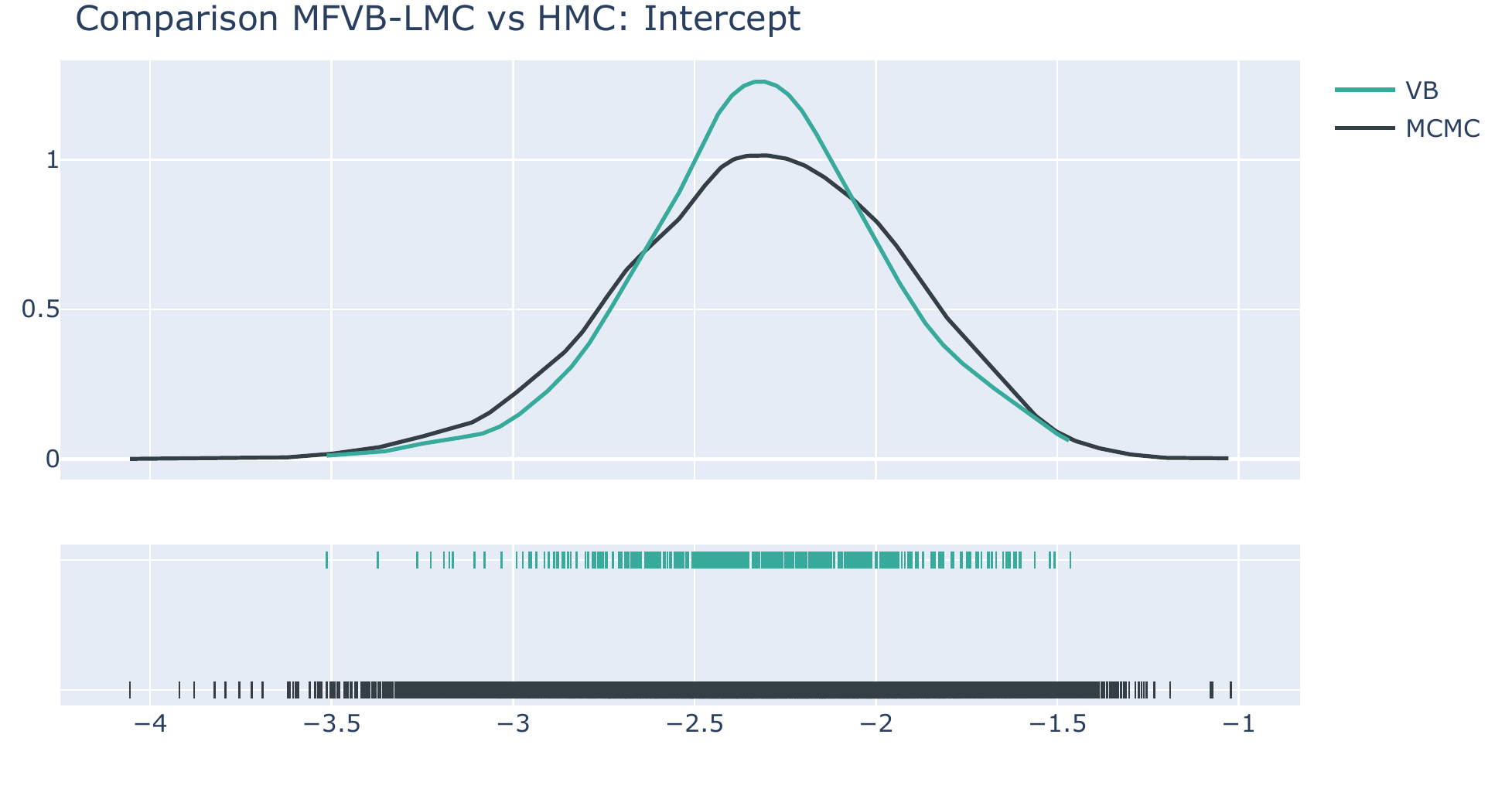}
  \caption{$\beta_0$}
\end{subfigure}%
\begin{subfigure}{.5\textwidth}
  \centering
  \includegraphics[width=\linewidth]{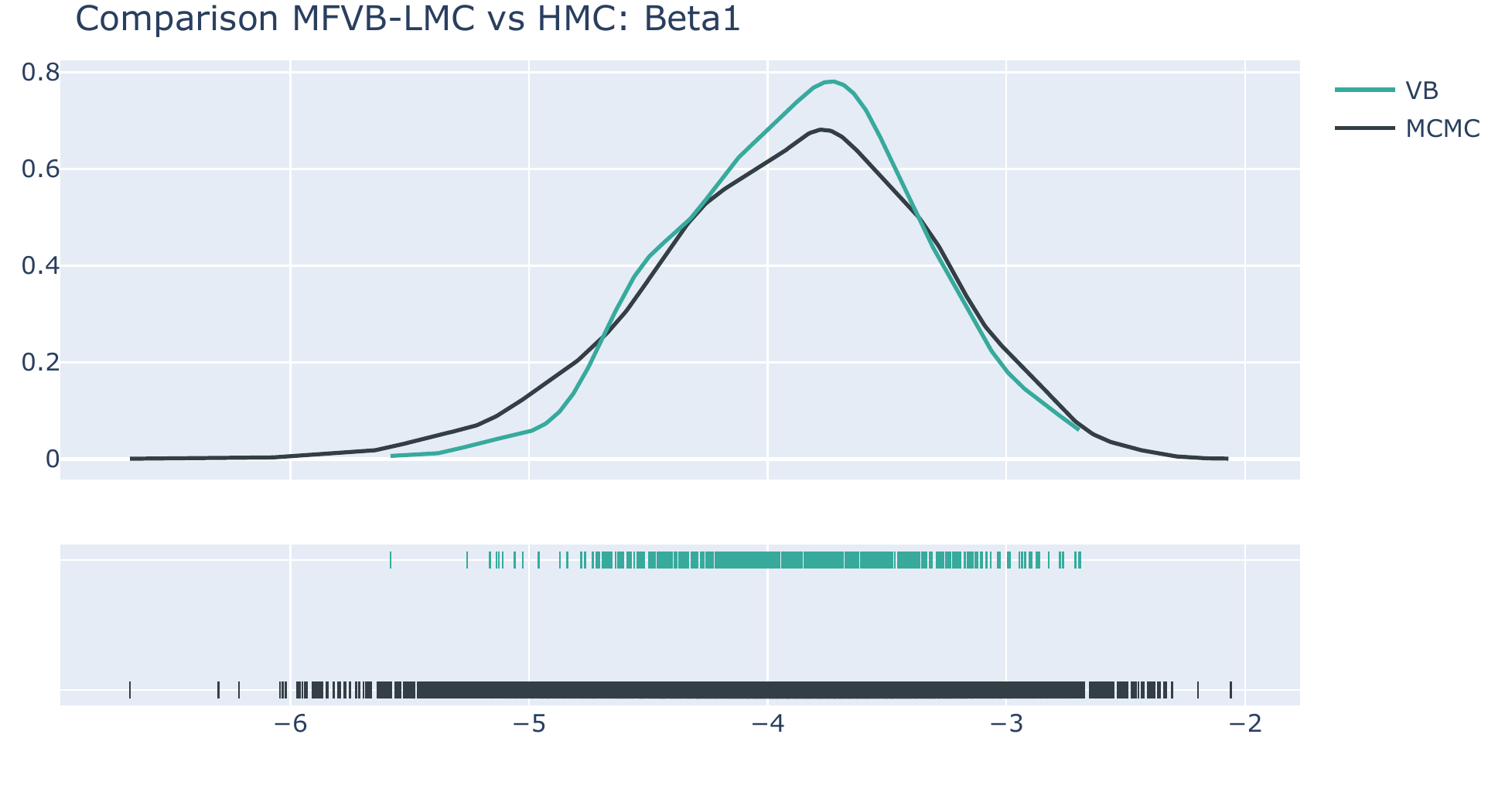}
  \caption{$\beta_1$}
\end{subfigure}
\begin{subfigure}{.5\textwidth}
  \centering
  \includegraphics[width=\linewidth]{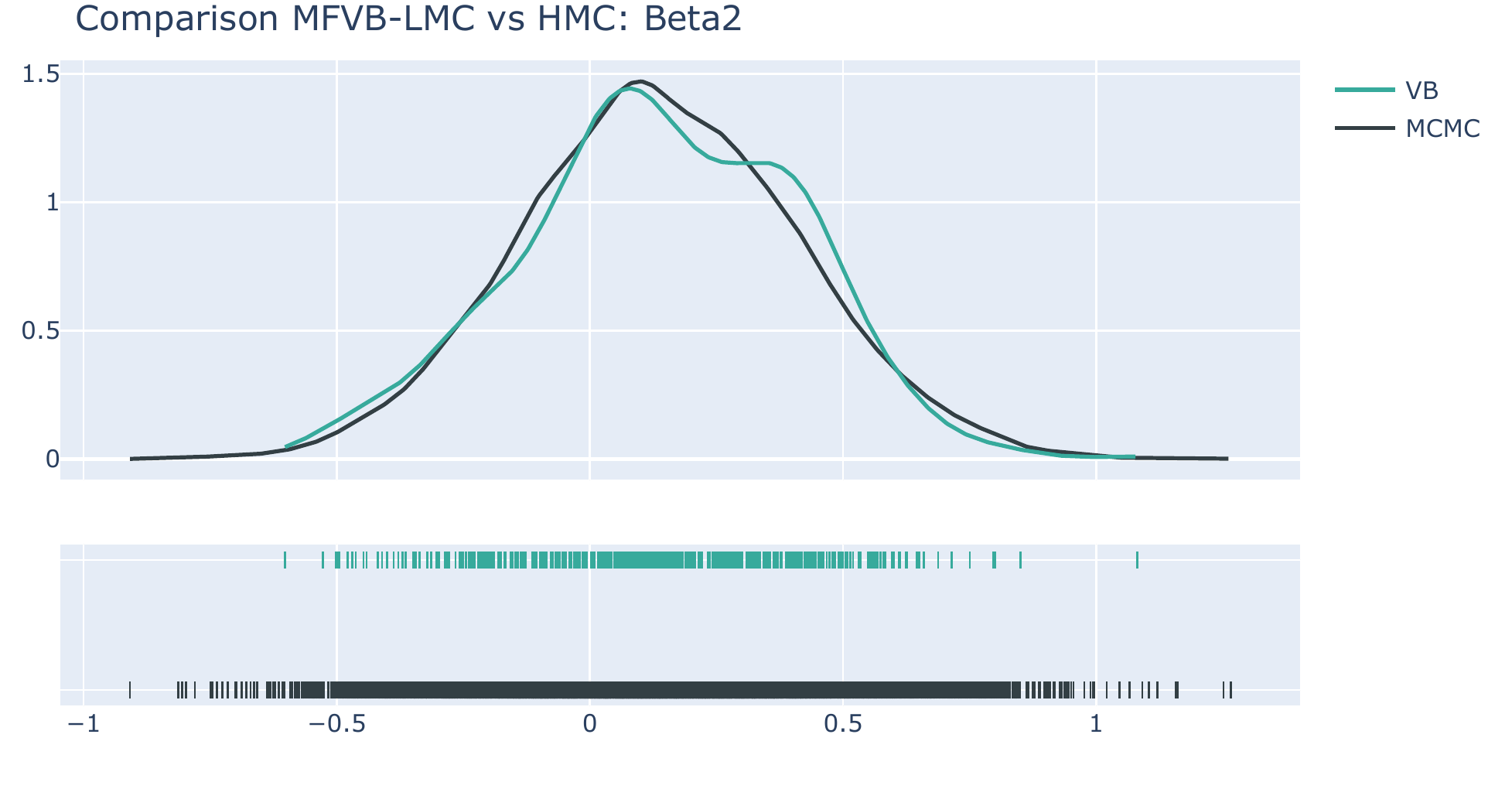}
  \caption{$\beta_2$}
\end{subfigure}%
\begin{subfigure}{.5\textwidth}
  \centering
  \includegraphics[width=\linewidth]{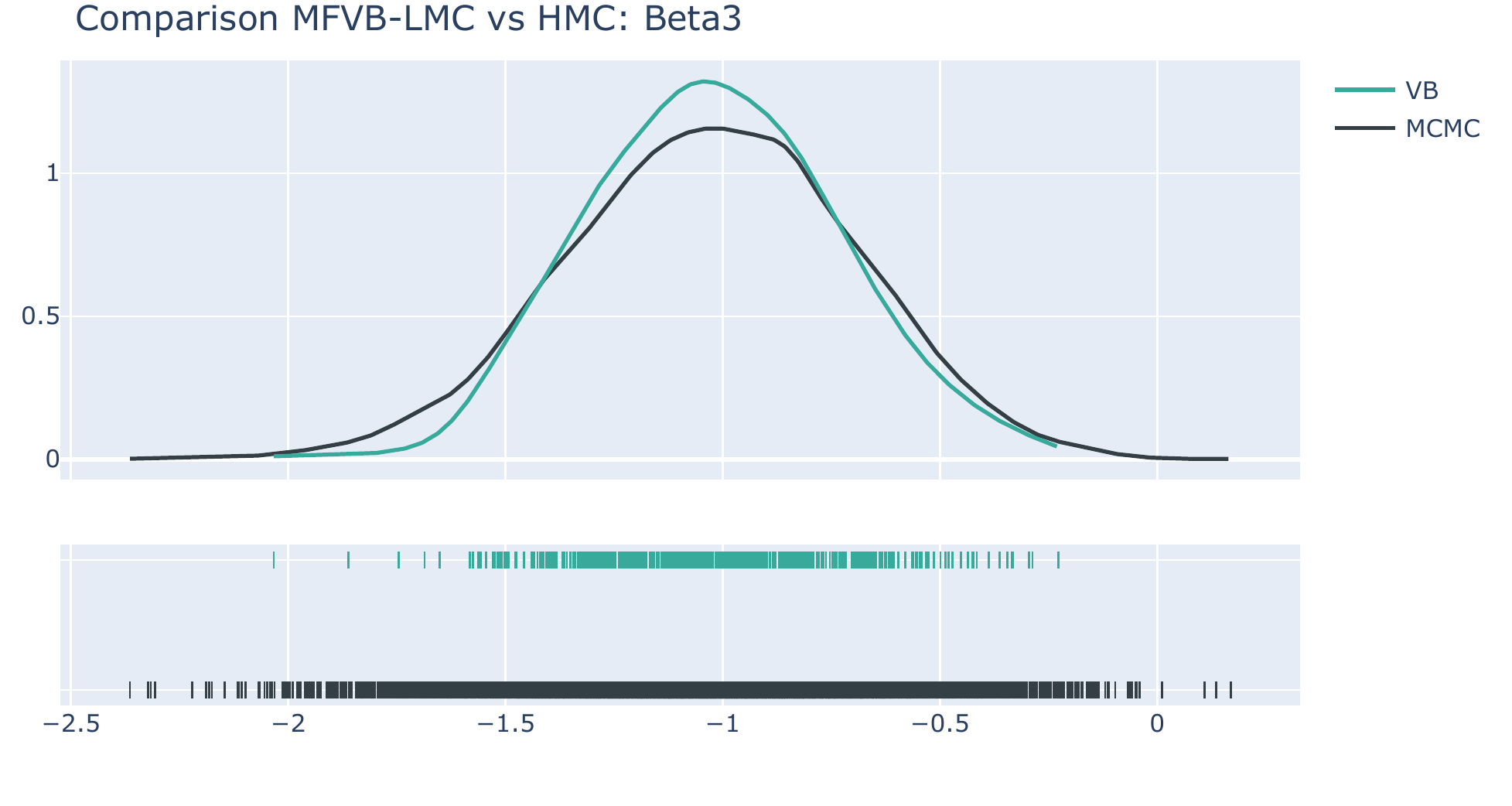}
  \caption{$\beta_3$}
\end{subfigure}
\caption{The estimated posterior density curves and rug plots of Bayesian logistic regression: the green is of the PMFVB and black of HMC (generated using PyMC).}
\label{fig:blogistic_hist}
\end{figure}

\subsection{Stochastic volatility}
This section applies the PMFVB method to Bayesian inference in the Stochastic Volatility (SV) model \citep{Taylor:SV1982}.
Let $\{y_t,\ t=1,2,...\}$ be an asset return time series.
The SV model is
\bea
y_{t}&\sim& N\big(0,e^{x_t}\big),\;\;t=1,2,...,T\label{eq:SV1}\\
x_{t}&\sim&N\big(\mu(1-\phi)+\phi x_{t-1},\sigma^2\big),\;\;t=2,...,T,\;\;x_1\sim N\big(\mu,\sigma^2/(1-\phi^2)\big),\label{eq:SV2}
\eea
with $\mu\in\SetR,\phi\in(-1,1)$ and $\sigma^2>0$ being the model parameters. Write $\theta=(\mu,\phi,\sigma^2)$.
Following \cite{kim1998stochastic}, we use the prior $\mu\sim N(0,\sigma_0^2)$ with $\sigma_0^2=10$, $\tau=(1+\phi)/2\sim\text{Beta}(a_0,b_0)$ with $a_0=20$ and $b_0=1.5$, and $\sigma^2\sim\text{inverse-Gamma}(\alpha_0,\beta_0)$ with $\alpha_0=2.5$ and $\beta_0=0.025$. 
It is challenging to perform Bayesian inference for the SV model because the likelihood $p(y|\theta)$ is a high-dimensional integral over the latent variables $x=x_{1:T}$.
A number of Bayesian methods are available for estimating the SV model including SMC$^2$ \citep{Chopin:JRSSB2013_SMC2,Gunawan:ES2022} and fixed-form Variational Bayes \citep{Tran:JCGS2019}.
However, to the best of our knowledge, MFVB has never been successfully used for this model.

To apply the PMFVB for SV, we use the following factorized variational distribution 
\beq\label{eq:SV PMFVB factorization}
q(\theta,x)=q(\mu)q(\sigma^2)q(\phi,x).
\eeq
This factorization leads to an analytical update for $q(\mu)$ and $q(\sigma^2)$, and we only need one LMC procedure to update $q(\phi,x)$, see Appendix A for the derivation.
One could also use
\beq
q(\theta,x)=q(\theta)q(x),
\eeq
but then two LMC procedures are needed to update $q(\theta)$ and $q(x)$ as $q(\theta)$ cannot be updated analytically.

We generate a return time series of $T=500$ observations from the SV model \eqref{eq:SV1}-\eqref{eq:SV2} with $\mu=1$, $\phi=0.8$ and $\sigma=0.5$.
Figure \ref{fig:SV} plots the posterior densities for $\theta$ estimated by the PMFVB and SMC methods. We use 500 particles in both methods. 
The CPU times taken by PMFVB and SMC were  8.2 and 29.3 minutes, respectively.
The running time for SMC depends on many factors such as the number of particles used in the particle filter for estimating the likelihood and the number of Markov moves. We select these numbers following their typical use in the literature; see, e.g., \cite{Gunawan:ES2022}.
Figure \ref{fig:SV} shows that the PMFVB estimates for $\mu$ and $\phi$ are almost identical to that of SMC except the estimate of $\sigma^2$, where PMFVB underestimates the posterior variance.
This is the well-known problem for the MFVB method due to the factorization \eqref{eq:SV PMFVB factorization} it imposes on the variational distribution. Section \ref{sec: discussion} suggests a possible solution.

\begin{figure}[H]
\centering
  \includegraphics[width=.7\linewidth]{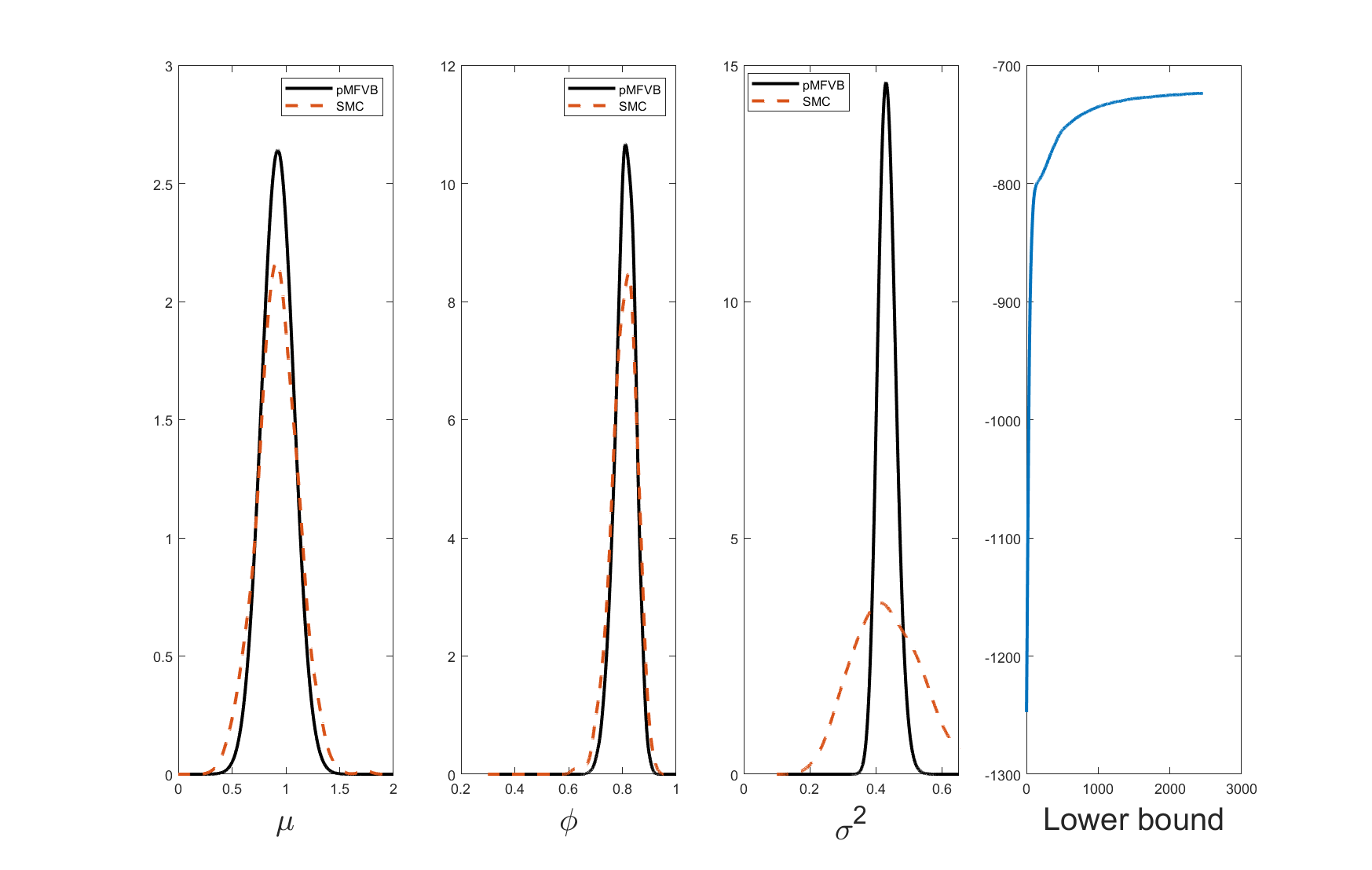}
\caption{The posterior density curves for the SV model parameters estimated by PMFVB (solid line) and SMC (dashed line). The right panel shows the smoothed lower bound in PMFVB over a window of size 50.}\label{fig:SV}
\end{figure}

\section{Particle MFVB for Bayesian neural networks }\label{sec: PMFVB-NN}
This section presents a variant of the PMFVB approach for Bayesian inference in big models like deep neural networks.
One of the most commonly used method for Bayesian inference in deep learning is perhaps the Stochastic Gradient Langevin Dynamics (SGLD) method of 
\cite{Welling.Teh:2011} and its variants.
Let $\theta\in\SetR^d$ be the model parameters and $\pi(\d\theta)$ their posterior distribution.
The SGLD algorithm is based on the discretised Langevin diffusion \eqref{eq: Langegin diffusion}
\beq
\theta^{(t+1)}=\theta^{(t)}+\frac{h}{2}\wh{\nabla\log\pi\big(\theta^{(t)}\big)}+\sqrt{h}\eta_t,\;\;\;\eta_t\stackrel{i.i.d}{\sim}N(0,I_d),
\eeq
where $\wh{\nabla\log\pi(\theta)}$ is an unbiased estimator of the gradient $\nabla\log\pi(\theta)$, often computed from a data mini-batch.
There is a large literature aiming at improving SGLD by exploiting the curvature structure of the log-target density function. For example, \cite{Li.et.al:2016} propose a preconditioned Stochastic Gradient Langevin Dynamics (pSGLD) that rescales the gradient of the log target density by a diagonal matrix learnt using the second moments of the previous gradients.
\cite{Kim.et.al:2022} introduce several SGLD schemes that add an adaptive drift to the noisy log-gradient estimator. 
See also \cite{Girolami:JRSSB2011} and \cite{Chen:ICML2014}.

\subsection{The algorithm}\label{sec: algorithm PMFVB-NN}
We introduce three refinements to make the PMFVB approach computationally efficient in big-data and complex-model situations.

First, we choose the updating block  randomly in each iteration $t$ and for each particle $i$. 
Let $\iota=\{1\leq j_1<j_2<...j_m\leq d\}$ be an index subset of size $m$ from $\{1,2,...,d\}$. We denote by $\theta(\iota)=(\theta_{j_1},...,\theta_{j_m})$ the sub-vector obtained from $\theta$ corresponding to the index set $\iota$; and by $\theta(\setminus\iota)$ the vector from $\theta$ after removing the components in $\theta(\iota)$.
For iteration $t$ and for each particle $i$, the index set $\iota_i$ is randomly selected (we suppress the dependence on $t$ for notational simplicity); the corresponding block of $m$ components $\theta_i^{(t)}(\iota_i)$ of $\theta_i^{(t)}$ is updated via LMC
\beq\label{eq:LMC components}
\theta_i^{(t+1)}(\iota_i)=\theta_i^{(t)}(\iota_i)+\frac{h}{2}\E_{q_{\theta(\setminus\iota)}^{(t)}}\Big[\nabla_{\theta(\iota_i)}\log\pi\big(\theta_i^{(t)}(\iota_i),\theta(\setminus\iota_i)\big)\Big]+\sqrt{h}\eta_{i},
\eeq
with $\eta_{i}\sim N_{m}(0,I)$. Here, $q_{\theta(\setminus\iota_i)}^{(t)}$ denotes the marginal distribution of the particles $\theta_i^{(t)}$ w.r.t. the $\theta(\setminus\iota_i)$ components.

Second, we approximate the expectation term $\E_{q_{\theta(\setminus\iota_i)}^{(t)}}[\cdot]$ in \eqref{eq:LMC components} by $\nabla_{\theta(\iota_i)}\log\pi\big(\theta_i^{(t)}(\iota_i),\bar\theta^{(t)}(\setminus\iota_i)\big)$,
where $\bar\theta^{(t)}$ is the sample mean of the particles $\{\theta_i^{(t)},i=1,...,M\}$. This leads to a significant reduction in computational time, compared to an alternative that averages over a subset of particles as in \eqref{eq: update X}.

Third, it is desirable to incorporate an adaptive SGLD scheme into the LMC update in PMFVB. Our work uses the ADAM-based adaptive-drift SGLD scheme of \cite{Kim.et.al:2022}. Algorithm \ref{Alg: particle MFVB-NN} summarises the method.

\begin{algorithm}[H]
\caption{Particle MFVB for neural networks}\label{Alg: particle MFVB-NN}
\begin{algorithmic}[1]
\Procedure{PMFVB-NN}{$h$, $\beta_1, \beta_2, a, \lambda, \epsilon$}
\State Input: step size $h$, smoothing weights $\beta_1, \beta_2$, tolerance $\epsilon>0$, scale factor $a$ and $\lambda>0$.
\State Initialise particles $\theta_i^{(0)} \sim p_0, ~i=1,...,M$, $m_0=0$, $v_0=0$, $t=0$.
\While{ \text{Not} $S(\epsilon)$}
\State $g_t \leftarrow m_t\oslash\sqrt{v_t+\lambda}$ and $\bar\theta^{(t)} \leftarrow \frac1M\sum_{i=1}^M\theta^{(t)}_i$ 
\For{$i$ in $[M]$}
\State  $\iota_i \leftarrow \{1\leq j_1<j_2<...j_m\leq d\}$, $j_i\in\{1,...,d\}$ are randomly selected.
\State Update: $\tilde{g}_t := \Big(\nabla_{\theta(\iota_i)}\log\pi\big(\theta_i^{(t)}(\iota_i),\bar\theta^{(t)}(\setminus\iota_i)\big)+ag_t(\iota_i)\Big)$
\beq
\begin{aligned}
\theta_i^{(t+1)}(\iota_i) &\leftarrow \theta_i^{(t)}(\iota_i) +\frac{h}{2}\tilde{g}_t + \sqrt{h}\eta_{i},~ \eta_{i}\sim N_{m}(0,I) \\ 
\theta_i^{(t+1)} & \leftarrow\big(\theta_i^{(t+1)}(\iota_i),\theta_i^{(t)}(\setminus\iota_i)\big)
\end{aligned}
\eeq
\State $g^{(t)}_i \leftarrow \nabla_{\theta}\log\pi(\theta_i^{(t+1)})$, new gradient
\State Update the adaptive drift
\bean
\bar g_t \leftarrow \frac1M\sum_{i=1}^Mg^{(t)}_i, &~~
\bar v_t \leftarrow \sqrt{\frac{1}{M}\sum_i(g^{(t)}_i-\bar g_t)\odot(g^{(t)}_i-\bar g_t)}\\
m_{t+1} \leftarrow \beta_1m_t+(1-\beta_1)\bar g_t, &~~
v_{t+1}\leftarrow\beta_2v_t+(1-\beta_2)\bar v_t.
\eean
\EndFor
\State $t=t+1$.
\EndWhile
\EndProcedure
\end{algorithmic}
\end{algorithm}

In Algorithm \ref{Alg: particle MFVB-NN}, $\oslash$ and $\odot$ denote the component-wise division and multiplication operators, respectively.

The implementation below follows \cite{Kim.et.al:2022} and sets $\beta_1=0.9$, $\beta_2=0.99$, $a=100$ and $\lambda=10^{-8}$. The block size $m$ is 10\% of the length $d$ of $\theta$. The step size $h$ is typically $0.001$.

\subsection{Applications}\label{sec: application PMFVB-NN}
\subsubsection{Bayesian neural networks for regression}
Consider the neural network regression model
\beq
y_i = \eta(x_i,w)+\epsilon_i,\;\;\epsilon_i\sim N(0,\sigma^2),\;\;i=1,...,n,
\eeq
where $\eta(x_i,w)$ is the output from a neural network with the input $x_i$ and weights $w=\{w_j,j=1,...,d_w\}$.
Neural networks are prone to overfitting and the Bayesian framework provides a principled  way for dealing with this problem by placing a regularization prior on the weights. We consider the following adaptive ridge-type regularization prior on the weights
\begin{equation}\label{eq:regpriorweights}
\begin{aligned}
w_j|\tau_j&\sim&N(0,\tau_j),\;\;j=1,...,d_w,\\
\tau_j&\sim&\text{Inverse Gamma}(\alpha_0,\beta_0).
\end{aligned}
\end{equation}
We set $\alpha_0=1$ and $\beta_0=0.01$ in all the examples below. For the variance $\sigma^2$, we use the improper prior $p(\sigma^2)\sim 1/\sigma^2$ \citep{Park:2008}.
The model parameters include $\theta=\big(w,\tau=\{\tau_j,j=1,...,d_w\},\sigma^2\big)$. Given the different roles of $w$, $\tau$ and $\sigma^2$, it is natural to use the following factorized variational distribution in PMFVB
\beq
q(\theta)=q(w)q(\tau)q(\sigma^2).
\eeq
It can be seen from \eqref{eq: optimal MFVB} that the optimal MFVB distribution for $\tau$ is $q(\tau)=\prod_jq(\tau_j)$, where $q(\tau_j)$ has an Inverse-Gamma density with scale $\alpha_0+1/2$ and rate $\beta_0+\langle w_j^2 \rangle/2$. Here, as common in the MFVB literature, $\langle\cdot\rangle$ denotes the expectation with respect to the variational distribution $q$.  
The optimal MFVB distribution for $\sigma^2$ is Inverse-Gamma with scale $n/2$ and shape $\frac12\langle\sum_{i=1}^n(y_i-\eta(x_i,w))^2\rangle$.
The terms $\langle w_j^2 \rangle$ and $\langle\sum_{i=1}^n(y_i-\eta(x_i,w))^2\rangle$ can be approximated using the $w$-particles.
The optimal MFVB distribution for the weights $w$ has the log-density
\beq\label{eq: q(w) term}
\log q(w)=-\frac12\sum_{j=1}^{d_w}\langle\frac{1}{\tau_j}\rangle w_j^2-\frac12\langle\frac{1}{\sigma^2}\rangle\sum_{i=1}^n(y_i-\eta(x_i,w))^2,
\eeq
which indicates that $q(\d w)$ cannot be updated analytically.
Note that the expectation terms $\langle\frac{1}{\tau_j}\rangle$ and $\langle\frac{1}{\sigma^2}\rangle$ can be computed analytically.
Based on \eqref{eq: q(w) term}, a Langevin MC procedure in Algorithm \ref{Alg: particle MFVB-NN} is used to approximate the distribution $q(\d w)$.

\subsubsection*{A simulation study} We simulate data from the following non-linear model \citep{Tran:JCGS2019} 
\beqn\label{eq:continuous model}
y=5 + 10x_1 + \frac{10}{x_2^2+1} + 5x_3x_4 + 2x_4 + 5x_4^2 + 5x_5 + 2x_6 + \frac{10}{x_7^2+1} + 5x_8x_9 + 5x_9^2 + 5x_{10} +\epsilon,
\eeqn
where $\epsilon\sim\N(0,1)$, $(x_1,...,x_{20})^\top$ are generated from a multivariate normal distribution with mean zero and covariance matrix $(0.5^{|i-j|})_{i,j}$; the last ten variables are not in the regression. The training data has 100,000 observations, the validation and test datasets each have 10,000 observations.

We compare the predictive performance of PMFVB with Gaussian VB of \cite{Tran:JCGS2019},
SGLD of \cite{Welling.Teh:2011}, preconditioned SGLD of \cite{Li.et.al:2016} and 
the ADAM-based adaptive drift SGLD of \cite{Kim.et.al:2022}.
We use 300 particles in PMFVB.
For the SGLD methods, one must first transform $\tau$ and $\sigma^2$ into an unconstrained space,
then apply the Langevin MC to jointly sample both the weights $w$, the transformed $\sigma^2$ and the latent (transformed) $\tau$.
The dimension of this LMC is double that of the LMC used in the PMFVB method where one needs to sample $w$ only. 
We set the tuning parameters in the SGLD methods following the suggestions in the papers proposing the methods.

The performance metrics include the best validation-data partial predictive score (PPS) and the test-data PPS,
\beqn
\text{PPS}=-\frac{1}{|D|}\sum_{(x,y)\in D}p(y|x,\wh\theta)
\eeqn
where $D$ is the validation and test data, respectively, $\wh\theta$ is the posterior mean estimate
of the model parameters. 
We also report the test-data MSE and the CPU running time (in minutes).
For each method, the Bayesian neural network model is trained until the validation PPS no longer decreases after 100 iterations. 
Table \ref{tab: simulation regression NN} summarizes the results,
which are based on 10 different runs for each algorithm.
The PMFVB method achieves the best predictive performance; it is also stable across different runs as reflected by the standard deviations.
In terms of the running time, the precondioned SGLD is the most efficient.
We note, however, that the PMFVB is parallelisable and its running time can be greatly reduced if multiple-core computers are used. 

\begin{table}[H]
\begin{center}
\begin{tabular}{lcccc}
\hline\hline
Method&Validation PPS& Test PPS & Test MSE & CPU\\
\hline
Gaussian VB         &1.2305 (0.0678)  &1.2513 (0.0714)    &4.3795 (0.6384)  &18.78\\ 
SGLD                &1.3751 (0.1060)  &1.3881 (0.1055)    &4.9997 (1.1608)  &2.74\\ 
Precondioned SGLD   &1.0527 (0.1274)  &1.0676 (0.1335)    &3.1517 (0.8158)  &{\bf 1.36}\\         
Adam SGLD           &1.1849 (0.0485)  &1.2021 (0.0487)    &3.7987 (0.4874)  &4.26\\
PMFVB               &{\bf 0.8239} (0.0642)  &{\bf 0.8598} (0.0640)    &{\bf 2.0631} (0.2603)  &16.29\\         
\hline\hline
\end{tabular}
\end{center}
\caption{Simulation data: Predictive performance in term of the partial predictive score (PPS), the mean squared error (MSE) and CPU time (in minutes), averaged over 10 different runs. The numbers in brackets are the standard deviation across the replicates. The best scores for each performance metric are highlighted in bold. The structure of the neural net is (20,20,20,1), i.e., one input layer of 20 units, two hidden layers each of 20 units and one output unit.}\label{tab: simulation regression NN}
\end{table}

\subsubsection*{The HILDA data}
The Household, Income and Labour Dynamics in Australia (HILDA) Survey\footnote{The HILDA Survey was conducted by the Australian Government Department of Social Services (DSS). The findings and views reported in this paper, however, are those of the authors and should not be attributed to the Australian Government, DSS, or any of DSS’ contractors or partners.} data consists of many household-based variables about economic and personal well-being.
We apply the Bayesian neural network model to predict the Income, using 43 covariates many of which are dummy variables used to encode the categorical variables. The dataset is randomly divided into a training set of 14,010 observations for fitting the model, a validation set of of 1751 observations for stopping,
and a test set of 1751 observations for final performance evaluation. We use a neural network with three hidden layers, each having 50 units, and use the same algorithm settings as in the previous simulation example.
Table \ref{tab: HILDA} summarizes the result which shows that the PMFVB algorithm achieves the best predictive performance
together with the highest stability (across different runs).

\begin{table}[H]
\begin{center}
\begin{tabular}{lcccc}
\hline\hline
Method&Validation PPS& Test PPS & Test MSE & CPU\\
\hline
Gaussian VB         &$-0.0957$ (0.0071)     &$-0.1474$ (0.0143)         &0.2739 (0.0078)    &10.62\\ 
SGLD                &$-0.1272$ (0.0035)     &$-0.1850$ (0.0051)         &0.2541 (0.0026)    &2.40\\ 
Precondioned SGLD   &$-0.1065$ (0.0141)     &$-0.1621$ (0.0190)         &0.2612 (0.0037)    &5.42\\
Adam SGLD           &$-0.1321$ (0.0035)     &$-0.1897$ (0.0022)         &0.2516 (0.0011)    &{\bf 1.57}\\
PMFVB               &${\bf -0.1331}$ (0.0019)&${\bf -0.1932}$ (0.0014)  &{\bf 0.2498} (0.0007)    &9.07\\
\hline\hline
\end{tabular}
\end{center}
\caption{HILDA data: Predictive performance in term of the partial predictive score (PPS), the mean squared error (MSE) and CPU time (in minutes), averaged over 10 different runs. The numbers in brackets are the standard deviation across the replicas. The best scores for each performance metric are highlighted in bold. The structure of the neural net is (43,50,50,50,1), i.e. three hidden layers each of 50 units.}\label{tab: HILDA}
\end{table}

\subsubsection{Bayesian neural networks for classification}
Consider the neural network binary classification model
\bean
y_i &\sim&\text{Binomial}\big(1,p(x_i,w)\big),\\
p(x_i,w) &=&\frac{1}{1+\exp\big(-\eta(x_i,w)\big)},\;\;i=1,...,n
\eean
where $\eta(x_i,w)$ is the output from a neural network with the input $x_i$ and weights $w=\{w_j,j=1,...,d_w\}$.
We use the same regularisation prior as in the regression (\ref{eq:regpriorweights}). The model parameters are $\theta=\big(w,\tau=\{\tau_j,j=1,...,d_w\}\big)$. We consider the following factorization in the PMFVB
\[q(\theta)=q(w)q(\tau).\]
The optimal MFVB distribution for $\tau$ is $q(\tau)=\prod_jq(\tau_j)$, where $q(\tau_j)$ has an Inverse-Gamma density with scale $\alpha_0+1/2$ and rate $\beta_0+\langle w_j^2 \rangle/2$. 
The term $\langle w_j^2 \rangle$ can be approximated from the $w$-particles.
The optimal MFVB distribution for the weights $w$ has the log-density
\beq\label{eq: q(w) term classification}
\log q(w)=-\frac12\sum_{j=1}^{d_w}\langle\frac{1}{\tau_j}\rangle w_j^2+\sum_{i=1}^n\Big(y_i\eta(x_i,w)-\log(1+e^{\eta(x_i,w)})\Big).
\eeq
The expectation term $\langle\frac{1}{\tau_j}\rangle$ can be computed analytically.
Based on \eqref{eq: q(w) term classification}, a Langevin MC procedure in Algorithm \ref{Alg: particle MFVB-NN} is used to approximate the distribution $q(\d w)$.

\subsubsection*{The census data}
The census dataset obtained from the U.S. Census Bureau is available on the UCI Machine Learning Repository. 
The task is to classify whether a person’s income is over \$50K per year, based on 14 covariates including
age, work class, race. Using dummy variables to represent the categorical variables, there are a total of 103 input variables in the Bayesian neural network model.
The full dataset of 45,221 observations is divided into a training set (53\%), validation set (14\%) and test set (33\%).
We use 200 particles in the PMFVB method.
Table \ref{tab: census data} summarizes the results
which show that the PMFVB obtains the best predictive performance.

\begin{table}[H]
\begin{center}
\begin{tabular}{lcccc}
\hline\hline
Method&Validation PPS & Test PPS&Test MCR & CPU\\
\hline
SGLD                &0.3153 (0.0046)        &0.4304 (0.0348)        &0.1990 (0.0033)        &3.83\\
Precondioned SGLD   &0.3133 (0.0012)        &0.4407 (0.0251)        &0.2034 (0.0051)        &3.47\\
Adam SGLD           &0.3109 (0.0011)        &0.4422 (0.0211)        &0.2097 (0.0079)        &{\bf 3.20}\\         
PMFVB               &{\bf 0.3052} (0.0008)  &{\bf 0.4217} (0.0069)  &{\bf 0.1958} (0.0030)  &18.30\\
\hline\hline
\end{tabular}
\end{center}
\caption{Census data: Predictive performance in term of the partial predictive score (PPS), miss-classification rate (MCR), and CPU time (in minutes), averaged over 10 different runs. The numbers in brackets are the standard deviation across the runs. The best scores for each performance metric are highlighted in bold. The structure of the neural net is (103,100,100,1).}\label{tab: census data}
\end{table}

\section{Discussion}\label{sec: discussion}
We propose a particle-based MFVB procedure for Bayesian inference, which extends the scope of classical MFVB, is widely applicable and enjoys attractive theoretical properties.
The new method can also be used for training Bayesian deep learning models.

The main limitation of MFVB methods including PMFVB is the use of factorized variational distributions, which might fail to capture the dependence structure between the blocks of variables.
This limitation can be mitigated using the reparametrisation method of \cite{tan2021use}.
Write the target distribution as $\pi(x,y)=\pi_x(x)\pi_{y|x}(y|x)$. Let $b(x)$ and $H(x)$ be the gradient and minus Hessian of $\log \pi_{y|x}(y|x)$ at some point $y=y_0$. Assume that $H(x)$ is positive definite, consider the transformation $\wt y=H(x)^{1/2}\big(y-\mu(x)\big)$ with $\mu(x)=H(x)^{-1}b(x)+y_0$. The joint density of $x$ and $\wt y$ is 
\[\wt\pi(x,\wt y)=|H(x)|^{-1/2}\pi_x(x)\pi_{y|x}\big(H(x)^{-1/2}\wt y+\mu(x)|x\big).\]
The motivation is that, if $\pi_{y|x}(y|x)\sim N\big(\mu(x),H(x)^{-1}\big)$, then $\wt y\sim N(0,I)$ and is independent of $x$.
In general, we can expect that the dependence between $x$ and $\wt y$ is reduced and is much less than the dependence between $x$ and $y$. \cite{tan2021use} considers this reparametrisation approach in Gaussian Variational Bayes and documents significant improvement in posterior approximation accuracy. Coupling the reparametrisation approach with PMFVB could lead to an efficient technique for Bayesian inference. This research is in progress.

\section*{Appendix A: Derivation for the SV example}
The joint posterior of $\theta$ and latent $x$ is
\bean
\pi(\theta,x)&\propto&p(\theta)p(x|\theta)p(y|x,\theta)\\
&=&p(\mu)p(\phi)p(\sigma^2)N\big(x_1|\mu,\frac{\sigma^2}{1-\phi^2}\big)\prod_{t=2}^TN\big(x_t|\mu(1-\phi)+\phi x_{t-1},\sigma^2\big)\prod_{t=1}^T N\big(y_t|0,e^{x_t}\big),
\eean
where $p(\mu)=N(\mu|0,\sigma^2)$, $p(\phi)\propto\big((1+\phi)/2\big)^{a_0-1}\big((1-\phi)/2\big)^{b_0-1}$, and
$p(\sigma^2)\propto(\sigma^2)^{-(1+\alpha_0)}\exp(-\beta_0/\sigma^2)$.
Using \eqref{eq: optimal MFVB}, the optimal MFVB distribution $q(\mu)$ is $N(\mu_q,\sigma_q^2)$ with $\mu_q=B/A$ and $\sigma_q^2=1/A$, where
\beqn
A=\frac{1}{\sigma_0^2}+\frac{\alpha_{\sigma^2}}{\beta_{\sigma^2}}\langle1-\phi^2\rangle+(T-1)\frac{\alpha_{\sigma^2}}{\beta_{\sigma^2}}\langle(1-\phi)^2\rangle
\eeqn
and
\beqn
B=\frac{\alpha_{\sigma^2}}{\beta_{\sigma^2}}\langle(1-\phi^2)x_1\rangle+(T-1)\frac{\alpha_{\sigma^2}}{\beta_{\sigma^2}}\langle(1-\phi)\sum_{t=2}^T(x_t-\phi x_{t-1})\rangle,
\eeqn
with $\alpha_{\sigma^2}$ and $\beta_{\sigma^2}$ given below.
Recall that $\langle\cdot\rangle$ denotes the expectation with respect to the variational distribution $q$.  
The optimal MFVB distribution $q(\sigma^2)$ is Inverse-Gamma$(\alpha_{\sigma^2},\beta_{\sigma^2})$, where $\alpha_{\sigma^2}=\alpha_0+T/2$ and
\beqn
\beta_{\sigma^2}=\beta_0+\frac12\Big\langle(1-\phi^2)\big[(x_1-\mu_q)^2+\sigma_q^2\big]+\sum_{t=2}^T\big[(x_t-\mu_q(1-\phi)-\phi x_{t-1})^2+(1-\phi)^2\sigma_q^2\big]\Big\rangle.
\eeqn
All the expectations $\langle\cdot\rangle$ in the expressions above are with respect to $q(\phi,x)$, which can be estimated from the $(\phi,x)$-particles. 
The logarithm of the optimal MFVB density for $(\phi,x)$ is
\bean
\log q(\phi,x)&=&\Big\langle \log p(\phi)+\log p(x_1|\theta)+\sum_{t=2}^T\log p(x_t|x_{t-1},\theta)+\sum_{t=1}^T\log p(y_t|x_t) \Big\rangle+C\\
&=&\Big\langle(a_0-1)\log(1+\phi)+(b_0-1)\log(1-\phi)+\frac{1}{2}\log(1-\phi^2)-\frac{1-\phi^2}{2\sigma^2}(x_1-\mu)^2\\
&&-\sum_{t=2}^T\frac{1}{2\sigma^2}\big(x_t-\mu(1-\phi)-\phi x_{t-1}\big)^2-\sum_{t=1}^T\big(\frac{x_t}{2}+\frac12y_t^2e^{-x_t}\big)\Big\rangle+C,
\eean
where $C$ is a constant independent of $\phi$ and $x$.
For the LMC step, we need the gradient $\nabla_\phi \log q(\phi,x)$ and $\nabla_x\log q(\phi,x)$.
\bean
\nabla_\phi \log q(\phi,x)&=&\frac{a_0-1}{1+\phi}-\frac{b_0-1}{1-\phi}-\frac{\phi}{1-\phi^2}+\phi\frac{\alpha_{\sigma^2}}{\beta_{\sigma^2}}\big[(x_1-\mu_q)^2+\sigma_q^2\big]\\
&&+\frac{\alpha_{\sigma^2}}{\beta_{\sigma^2}}\sum_{t=2}^T\big[(x_{t-1}-\mu_q)(x_{t}-\mu_q)-\phi(x_{t-1}-\mu_q)^2+(1-\phi)\sigma_q^2\big],
\eean
\beqn
\nabla_{x_1}\log q(\phi,x)=-\frac{\alpha_{\sigma^2}}{\beta_{\sigma^2}}(1-\phi^2)(x_1-\mu_q)-\frac12+\frac{y_1^2}{2}e^{-x_1}+\phi\frac{\alpha_{\sigma^2}}{\beta_{\sigma^2}}\big(x_{2}-(1-\phi)\mu_q-\phi x_1\big),
\eeqn
\beqn
\nabla_{x_t}\log q(\phi,x)=-\frac12+\frac{y_t^2}{2}e^{-x_t}-
\frac{\alpha_{\sigma^2}}{\beta_{\sigma^2}}\big(x_{t}-(1-\phi)\mu_q-\phi x_{t-1}\big)
+\phi\frac{\alpha_{\sigma^2}}{\beta_{\sigma^2}}\big(x_{t+1}-(1-\phi)\mu_q-\phi x_{t}\big),
\eeqn
for $t=2,...,T-1$, and finally,
\beqn
\nabla_{x_T}\log q(\phi,x)=-\frac12+\frac{y_T^2}{2}e^{-x_T}-
\frac{\alpha_{\sigma^2}}{\beta_{\sigma^2}}\big(x_{T}-(1-\phi)\mu_q-\phi x_{T-1}\big).
\eeqn

\section*{Appendix B: Technical proofs}
\begin{proof}[Proof of Theorem \ref{the: theorem 1}]
Consider two measures in $\Q=\W_2(\X)\otimes \W_2(\Y)$: $q^{(1)}(\d x\times\d y)=q_x^{(1)}(\d x)q_y^{(1)}(\d y)$ and $q^{(2)}(\d x\times\d y)=q_x^{(2)}(\d x)q_y^{(2)}(\d y)$. Then,
\beq\label{eq: fact 1}
W_2^2(q^{(1)},q^{(2)})=W_2^2(q_x^{(1)},q_x^{(2)})+W_2^2(q_y^{(1)},q_y^{(2)}).
\eeq
With a generic measure $q(\d x\times\d y)=q_x(\d x)q_y(\d y)\in\Q$, write $F(q)$ as
\beq
F(q) = F_x(q_x)+F_y(q_y)+F_{xy}(q)
\eeq 
where
\[F_x(q_x)=\int_{\X}q_x(x)\log q_x(x)\d x,\;\;\;F_y(q_y)=\int_{\Y}q_y(y)\log q_y(y)\d y\]
and 
\[F_{xy}(q)=\int_{\X\times\Y}\big(-\log\pi(x,y)\big)q(x,y)\d x\times\d y.\]
To show that $F(q)$ is lower semi-continuous (l.s.c), consider a sequence of measures $\{q^n(\d x\times\d y)=q_x^n(\d x)q_y^n(\d y)\}_{n\geq1}\subset\W_2(\Q)$ weakly converging to $q^*(\d x\times\d y)=q_x^*(\d x)q_y^*(\d y)$, i.e.
\[W_2(q^n,q^*)\to0,~ \text{ as } n \to \infty.\]
Then \eqref{eq: fact 1} implies that 
\[W_2(q_x^n,q_x^*)\to0,\;\;\;\;\text{ and }\;\;\;\;\;W_2(q_y^n,q_y^*)\to0,\]
hence
\[q_x^n\stackrel{w}{\longrightarrow}q_x^*,\;\;\;\;\text{ and }\;\;\;\;\;q_y^n\stackrel{w}{\longrightarrow}q_y^*.\]
By Proposition 7.7 of \cite{Santambrogio:OTbook}, $F_x(q_x)$ and $F_y(q_y)$ are l.s.c; hence we have
\beq\label{eq: fact 2}
\liminf_n F_x(q_x^n)\geq F_x(q_x^*),\;\;\;\;\text{ and }\;\;\;\;\;\liminf_n F_y(q_y^n)\geq F_y(q_y^*).
\eeq
As $\log\pi(x,y)$ is continuous, and hence bounded, on the compact set $\X\times\Y$, by the definition of weak convergence, we have that
\beq\label{eq: fact 3}
\lim_n F_{xy}(q^n)= F_{xy}(q^*).
\eeq
From \eqref{eq: fact 2}-\eqref{eq: fact 3}, 
\beq\label{eq: fact 4}
\liminf_n F(q^n)\geq F(q^*),
\eeq
proving that $F(q)$ is l.s.c.

We now show that $F(q)$ is convex. Consider two measures $q^{(1)}, q^{(2)}\in\Q$ and any $t\in(0,1)$.
Because $f(z)=z\log(z)$ is convex, $F_x(q_x)=\int f(q_x)\d x$ and $F_y(q_y)=\int f(q_y)\d y$ are convex. Also, $F_{xy}(q)$ is linear and hence convex. Therefore,
\bean
F(tq^{(1)}+(1-t)q^{(2)}) &=& F_x(tq_x^{(1)}+(1-t)q_x^{(2)})+F_y(tq_y^{(1)}+(1-t)q_y^{(2)})+F_{xy}(tq^{(1)}+(1-t)q^{(2)})\\
&\leq& t\big(F_x(q_x^{(1)})+F_y(q_y^{(1)})+F_{xy}(q^{(1)})\big)+(1-t)\big(F_x(q_x^{(2)})+F_y(q_y^{(2)})+F_{xy}(q^{(2)})\big)\\
&=&tF(q^{(1)})+(1-t)F(q^{(2)}).
\eean
\end{proof} 

\begin{proof}[Proof of Corollary \ref{cor: unique solution}] 
As $\X$ and $\Y$ are compact, by the Prokhorov theorem, $\P(\X)$ and $\P(\Y)$ are compact (w.r.t. to the weak convergence, and also w.r.t the Wasserstein metric).
As $\P_2(\X)\subset\P(\X)$, for any sequence of measures $\{\mu_n\}$ in $\P_2(\X)$, there must exist a subsequence $\{\mu_{n_k}\}$ weakly converging to some measure $\mu\in\P(\X)$.
As $\X$ is compact, $\int_{\X}|x|^2\mu(\d x)<\infty$, hence $\mu\in\P_2(\X)$. This implies that $\P_2(\X)$ is compact.
Similarly, $\P_2(\Y)$ is compact, and therefore the product space $\Q=\W_2(\X)\otimes\W_2(\Y)$ is compact.
Recall that $\W_2(\X)$ (res. $\W_2(\Y)$) is $\P_2(\X)$ (res. $\P_2(\Y)$) equipped with the Wasserstein distance.
From Theorem \ref{the: theorem 1}, $F(q)$ is l.s.c on the compact space $\Q$, by the Weierstrass theorem, there exists $q^*\in\Q$ such that $F(q^*)=\min\{F(q):q\in\Q\}$.
The uniqueness of $q^*$ is implied by the fact that $F(q)$ is convex.
\end{proof}

\begin{proof}[Proof of Theorem \ref{the: Convergence of the particle MFVB}]
Given $q_y^{(t)}$, define
\[q_x^*(x)=\exp\Big(\E_{q_y^{(t)}}\big[\log\pi(x,y)\big]+C(q_y^{(t)})\Big),\]
where $C(q_y^{(t)}$ is the normalising constant. We have that
\bean
F(q^{(t)})=\KL(q_x^{(t)}q_y^{(t)}\|\pi)&=&\int q_x^{(t)}(x)q_y^{(t)}(y)\log\frac{q_x^{(t)}(x)q_y^{(t)}(y)}{\pi(x,y)}\d x\times\d y\\
&=&\int q_x^{(t)}(x)\Big(\log q_x^{(t)}(x)-\E_{q_y^{(t)}}\big(\log\pi(x,y)\big)\Big)\d x+E(q_y^{(t)})\\
&=&\KL(q_x^{(t)}\|q_x^*)+E(q_y^{(t)})+C(q_y^{(t)})
\eean
with $E(q_y^{(t)})=\int q_y^{(t)}(y)\log q_y^{(t)}(y)\d y$.
If the step size $h_x$ is sufficiently small, Lemma \ref{lem:Cheng and Bartlett} guarantees that 
\[\KL(q_x^{(t+1)}\|q_x^*)\leq \KL(q_x^{(t)}\|q_x^*);\]
hence
\beq\label{eq: F inequality 1}
F(q_x^{(t+1)}q_y^{(t)}) \leq F(q_x^{(t)}q_y^{(t)})=F(q^{(t)}).
\eeq
Given $q_x^{(t+1)}$, define
\[q_y^*(y)=\exp\Big(\E_{q_x^{(t+1)}}\big[\log\pi(x,y)\big]+C(q_x^{(t+1)})\Big).\]
We have that
\[F(q_x^{(t+1)}q_y^{(t)})=\KL(q_y^{(t)}\|q_y^*)+E(q_x^{(t+1)})+C(q_x^{(t+1)}).\]
Lemma \ref{lem:Cheng and Bartlett} guarantees that 
\[\KL(q_y^{(t+1)}\|q_y^*)\leq \KL(q_y^{(t)}\|q_y^*)\]
and hence
\beq\label{eq: F inequality 2}
F(q^{(t+1)})=F(q_x^{(t+1)}q_y^{(t+1)})\leq F(q_x^{(t+1)}q_y^{(t)}).
\eeq
From \eqref{eq: F inequality 1}-\eqref{eq: F inequality 2}, $F(q^{(t)})$ is reduced over $t$.
By Corollary \ref{cor: unique solution}, $q^{(t)}$ must converge to the unique  minimizer $q^*$ of $F(q)$.
\end{proof}

\begin{proof}[Proof of Theorem \ref{the: posterior consistency}]
Under the conditions (A1) and (A2), by Theorem 1 of \cite{ZhangGao:2019},
\beq\label{eq:ChapterVB_theory_convergence_rate}
\E_{p_0^{(n)}}\Big[\E_{q_n^*}\Big(\|\theta-\theta_0\|_2^2\Big)\Big]=O(\varepsilon_n^2+\gamma_n^2)
\eeq
with
\beq\label{eq:ChapterVB_theory_gamma2_n}
\gamma_n^2=\frac{1}{n}\min_{q\in\mathcal Q}\E_{p_0^{(n)}}\Big[\KL\big(q\|\pi_n\big)\Big].
\eeq 
Denote by
\[p^{(n)}(X^{(n)})=\int p_\theta^{(n)}(X^{(n)})\pi_0(\d\theta)\]
the marginal likelihood. For any $q\in\Q$, we have
\bea\label{eq:gamma 1}
\gamma_n^2&\leq&\frac{1}{n}\E_{p_0^{(n)}}\Big[\int\log\frac{q(\theta)p^{(n)}(X^{(n)})}{\pi_0(\theta)p_\theta^{(n)}(X^{(n)})}q(\d\theta)\Big]\notag\\
&=&\frac{1}{n}\KL(q\|\pi_0)+\frac{1}{n}\int \Big(p_0^{(n)}(X^{(n)})\int\log\frac{p^{(n)}(X^{(n)})}{p_\theta^{(n)}(X^{(n)})}q(\d\theta)\Big)\d X^{(n)}\notag\\
&=&\frac{1}{n}\KL(q\|\pi_0)+\frac{1}{n}\E_{q}\Big(\int p_0^{(n)}(X^{(n)})\log\frac{p_0^{(n)}(X^{(n)})}{p_\theta^{(n)}(X^{(n)})}\d X^{(n)}-\int p_0^{(n)}(X^{(n)})\log\frac{p_0^{(n)}(X^{(n)})}{p^{(n)}(X^{(n)})}\d X^{(n)}\Big)\notag\\
&=&\frac{1}{n}\KL(q\|\pi_0)+\frac{1}{n}\E_{q}\Big(\KL\big(p_0^{(n)}\|p_\theta^{(n)}\big)-\KL\big(p_0^{(n)}\|p^{(n)}\big)\Big)\notag\\ 
&\leq&\frac{1}{n}\KL(q\|\pi_0)+\frac{1}{n}\E_{q}\Big(\KL\big(p_0^{(n)}\|p_\theta^{(n)}\big)\Big).
\eea
Select $q(\theta)=N(\theta_0,1/nI_d)$, and estimate the first term in \eqref{eq:gamma 1}. 
\beqn
\frac{1}{n}\KL(q\|\pi_0)\leq\frac{1}{n}|\E_q\log q(\theta)|+\frac{1}{n}|\E_q\log\pi_0(\theta)|.
\eeqn
We have that
\[\frac{1}{n}|\E_q\log q(\theta)|\leq\frac{1}{n}\big(\frac{d}{2}\log(2\pi)+\frac{d}{2}\log n+\frac{d}{2}\big)=o(1).\]
By Assumption (A3),
\bean
\frac{1}{n}|\E_q\log\pi_0(\theta)|&\leq&\frac{1}{n}\Big(|\log\pi_0(\theta_0)|+\E_q|\log\pi_0(\theta)-\log\pi_0(\theta_0)|\Big)\\
&\leq&\frac{1}{n}\Big(|\log\pi_0(\theta_0)|+C_3\E_q\|\theta-\theta_0\|_2\Big)\\
&\leq&\frac{1}{n}\Big(|\log\pi_0(\theta_0)|+C_3\sqrt{\frac{d}{n}}\Big)=o(1).
\eean
Therefore,
\beq\label{eq:gamma evaluation 1}
\frac{1}{n}\KL(q\|\pi_0)=o(1).
\eeq
We now estimate the second term in \eqref{eq:gamma 1}. 
By Assumption (A3),
\bea\label{eq:gamma evaluation 2}
\frac{1}{n}\E_{q}\Big(\KL\big(p_0^{(n)}\|p_\theta^{(n)}\big)\Big)&\leq&\frac{1}{n}\E_{q}\Big(\E_{p_0^{(n)}}\big|\log p_\theta^{(n)}(X^{(n)})-\log p_{\theta_0}^{(n)}(X^{(n)})\big|\Big)\notag\\
&\leq&\frac{1}{n}C_5\E_{q}\big(\|\theta-\theta_0\|_2\big)\notag\\
&\leq&\frac{1}{n}C_5\sqrt{\frac{d}{n}}=o(1).
\eea
Equations \eqref{eq:gamma 1}, \eqref{eq:gamma evaluation 1} and \eqref{eq:gamma evaluation 2} imply that $\gamma_n^2=o(1)$.
Therefore, from \eqref{eq:ChapterVB_theory_convergence_rate} and noting that $\varepsilon_n^2=o(1)$, we have that 
\beqn
\E_{p_0^{(n)}}\Big[\E_{q_n^*}\Big(\|\theta-\theta_0\|_2^2\Big)\Big]=o(1).
\eeqn
By Markov's inequality
\beqn
\E_{p_0^{(n)}}\Big[q_n^*\big(\|\theta-\theta_0\|_2^2>\epsilon\big)\Big]\leq \frac{\E_{p_0^{(n)}}\Big[\E_{q_n^*}\Big(\|\theta-\theta_0\|_2^2\Big)\Big]}{\epsilon}\to 0,
\eeqn
implying
\beqn
q_n^*\Big(\|\theta-\theta_0\|_2^2>\epsilon\Big)\stackrel{n\to\infty}{\longrightarrow}0,\;\;\;\;p^{(n)}_0-a.s. 
\eeqn
\end{proof}

\bibliographystyle{apalike}
\bibliography{references}

\end{document}